\documentclass[a4paper,12pt]{article}
\usepackage[14pt]{extsizes} 
\usepackage[english]{babel}
\usepackage[mathscr]{eucal}
\usepackage{amsmath,amsfonts,amssymb,amsthm}
\usepackage{times}
\usepackage{hyperref}

\numberwithin{equation}{section} 

\newtheorem{thm}{Theorem}[section]

 \newtheorem{prop}[thm]{Proposition}
 
 \newtheorem{definition}[thm]{Definition}
 
\renewcommand{\tilde}{\widetilde}

\renewcommand{\simeq}{\cong}
\newcommand{\bref}[1]{\textbf{\ref{#1}}}

\newcommand{\p}[1]{|#1|}
\newcommand{\gh}[1]{\mathrm{gh}(#1)}

\newcommand{\dx}{\mathrm{d}_X}

\renewcommand{\d}{\partial}

\newcommand{\tensor}{\otimes}

\renewcommand{\geq}{\,{\geqslant}\,}
\renewcommand{\leq}{\,{\leqslant}\,}

\newcommand{\binner}[2]{%
  {\langle}\kern-4.15pt{\langle}#1{,}\,#2{\rangle}\kern-4.15pt{\rangle}}
\newcommand{\commut}[2]{[#1{,}\,#2]}

\newcommand{\pb}[2]{\left\{{}#1{},{}#2{}\right\}}

\newcommand{\half}{\mathchoice{%
    \ffrac{1}{2}}{\frac{1}{2}}{\frac{1}{2}}{\frac{1}{2}}}

\newcommand{\ffrac}[2]{\raisebox{.5pt}%
  {\footnotesize$\displaystyle\frac{#1}{#2}$}\kern1pt}

\newcommand{\dl}[1]{\mathchoice{\ffrac{\d}{\d #1}}{\frac{\d}{\d #1}}{\ffrac{\d}{\d #1}}{\ffrac{\d}{\d #1}}}

\newcommand{\st}[2]{{\overset{#1}{#2}}}

\newcommand{\ddl}[2]{\ffrac{\d #1}{\d #2}}

\newcommand{\Liealg}{\mathfrak} 
\newcommand{\algg}{\Liealg{g}}

\newcommand{\cC}{\mathcal{C}}

\newcommand{\fR}{\mathbb{R}}

\newcommand{\fZ}{\mathbb{Z}}

 \def\cI{\mathcal{I}}
 
\def\cK{\mathcal{K}}
 
\def\tr{{\rm Tr}}

\newcommand\blfootnote[1]{%
  \begingroup
  \renewcommand\thefootnote{}\footnote{#1}%
  \addtocounter{footnote}{-1}%
  \endgroup
}

\usepackage[left=3cm,right=2cm,
    top=3cm,bottom=3cm,bindingoffset=0cm]{geometry}

\newcommand{\comm}[1]{}
\newcommand{\ev}{\mathrm{ev}}
\newcommand{\ferm}[1]{\epsilon(#1)}
\usepackage{authblk}
\usepackage{mathtools}
\usepackage[dvipsnames]{xcolor}
\usepackage{comment}

\title{Topological deformations of local gauge theories: a presymplectic \\ BV-AKSZ approach}

\author[1]{Ivan Dneprov}

\author[2]{Alexey Donchenko}

\author[1,2,$\dagger$]{Maxim Grigoriev}

\affil[1]{\textsl{Service de Physique de l'Univers, Champs et Gravitation, \protect\\ Universit\'e de Mons, 20 place du Parc, 7000 Mons, 
Belgium \vspace{5pt}}}

\affil[2]{\textsl{ Institute for Theoretical and Mathematical Physics,\protect\\
  Lomonosov Moscow State University, 119991 Moscow, Russia  \vspace{5pt}}}

\date{}

\begin{document}

\maketitle
\begin{abstract}

The same local degrees of freedom can be described by systems with different gauge structure and field content. Standard examples are the duality between the scalar field and the 2-form gauge field in four dimensions, and unimodular gravity, where the cosmological constant becomes a global degree of freedom. In a precise sense, the spaces of gauge-inequivalent formal solutions of such systems differ by a finite-dimensional sector. Moreover, because their gauge sectors are different, their consistent interactions and their sensitivity to the global structure of spacetime may differ substantially. We show that such systems are naturally related by what we call topological deformations: gauging a symmetry of generic degree couples the system to a BF-type sector without introducing new local degrees of freedom. The construction is carried out systematically within the presymplectic BV-AKSZ approach, where it amounts to a modification of the underlying presymplectic gauge PDE. Gauging degree-0 symmetries identifies solutions related by global symmetries, which reproduces the scalar–2-form and principal chiral model–Freedman–Townsend dualities. Gauging degree-1 symmetries, i.e.\ deformations, gives rise to new global degrees of freedom, such as the cosmological constant in unimodular gravity.

\end{abstract}

\blfootnote{${}^{\dagger}$ Also at Lebedev Physical Institute, Moscow, Russia}
\thispagestyle{empty}
\newpage

\setcounter{tocdepth}{2}

\tableofcontents

\section{Introduction}

We often meet pairs of field theories which are close to being equivalent in the sense that they have the same local degrees of freedom but their gauge structure is different and the moduli spaces of solutions may differ by a finite-dimensional factor. There are at least two known types of such pairs, which we refer to as pairs of deformation type and pairs of duality type. In pairs of deformation type one theory has a global conserved charge while the other one is merely the first one with a fixed value of this charge. An elegant example is provided by unimodular gravity, where the cosmological constant is promoted to a dynamical field whose equations of motion imply that it is constant and hence is a global conserved charge. The system obtained by setting this ``cosmological field'' to a particular value is locally equivalent to the usual general relativity. Pairs of duality type are the so-called dual formulations of the same local degrees of freedom. The standard example is the duality between the massless scalar field and the $2$-form gauge field in $4$-dimensional spacetimes.

In both of the above situations such a pair can be constructed by applying to a given gauge system a certain systematic procedure which in this work is referred to as gauging, though it is often called the flat gauging. In a very precise sense this procedure makes a generic-degree global symmetry into a gauge symmetry without adding local degrees of freedom.

Pairs of deformation type can be understood as resulting from the gauging of deformations which, in turn, may be regarded as ghost degree-$1$ symmetries (also known as form degree $-1$ symmetries). In this case the space of inequivalent solutions becomes larger because the corresponding ghost field(s) give a new global degree of freedom. In contrast, pairs of duality type result from gauging usual (i.e.\ degree-zero) symmetries. This reduces the space of gauge-inequivalent solutions because the configurations which were related by the global symmetry transformation are now related by the gauge one and hence are to be considered equivalent. Of course, if global geometry effects are not disregarded the relation can be much more subtle.

To be a little bit more precise let us give some more details on what the gauging of a deformation amounts to, when applied to a Lagrangian field theory without gauge symmetry. Consider a local field theory described by an action functional of the form
\begin{equation} \label{act+const}
    S[\phi; \lambda) = S_0[\phi] + \lambda_0 \int_X O\,,
\end{equation}
where $X$ is the spacetime manifold of dimension $n$ and $\lambda_0 \in \fR$ is a constant parameter. This system can be regarded as a consistent deformation of $S_0[\phi]$ in which case $\lambda_0$ is interpreted as the deformation parameter (coupling constant). In this situation the gauging of this deformation is an extended system, where $\lambda_0$ is promoted to a scalar field $\lambda$ and supplemented with its associated polymomentum $B$ such that the equations of motion imply $d\lambda=0$ and the initial dynamics is recovered if one restricts to the solutions with $\lambda=\lambda_0$. More precisely, one introduces scalar field $\lambda$ and the $n-1$-form field $B$ and considers the following action: 
\begin{equation} \label{act+const-gauged}
    S'[\phi, \lambda, B] = S_0[\phi] + \int_X \lambda(O + dB)\,.
\end{equation}
Here and in what follows we omit the wedge symbol for the exterior product of forms. The corresponding equations of motion read as: 
\begin{equation}
    \begin{gathered}
        d \lambda = 0\,, \\
        O + dB = 0\,, \\
        \delta S_0/ \delta \phi + \lambda \delta O/ \delta \phi = 0\,.
    \end{gathered}
\end{equation}
The first equation states that the field $\lambda$ is constant on-shell. Inserting it into the equations for $\phi$ gives back equations of motion which coincide with equations deduced from~\eqref{act+const} for $\lambda = \lambda_0$. If spacetime does not have nontrivial de-Rham cohomology in degree $n$ and $n-1$, the second equation can always be solved locally and its solution is unique modulo the natural gauge symmetry $B \to B+d\mu$. It is clear that for a sufficiently small interval containing a generic $\lambda_0 \in \fR$, the space of gauge-inequivalent solutions of the resulting system is the product of the initial one with the interval. In general the structure can be more involved as the initial solution space may differ drastically for different values of $\lambda$. This mechanism is known in the literature in one or another version for quite some time, see e.g.~\cite{Henneaux:1989zc,Bergshoeff:1996ui} and was recently discussed in some details in e.g.~\cite{Hajian:2023bhq,Tanizaki:2019rbk}.

As for the gauging usual global symmetries, the prototypical example is as follows. Let $S_0[\phi]$ be a first-order action functional invariant under a global symmetry Lie group $G$. At the infinitesimal level this means that there exists an action of the corresponding Lie algebra $\algg$ on the space of fields that preserves $S_0[\phi]$ and we denote by $\delta_\alpha \phi=T_\alpha(\phi)$ the corresponding transformations. By the Noether theorem such symmetries give rise to a conserved $n-1$ form $J$ with values in $\algg^*$, which we assume not to depend on derivatives of $\phi$ for the sake of simplicity. The standard "gauging" of such a symmetry corresponds to introducing a background field $A$, a 1-form valued in $\algg$, and replacing the starting point action $S_0[\phi]$ with
\begin{equation}
    S[\phi; A) = S_0[\phi] + \int_X \tr (AJ)\,,
\end{equation}
which is gauge invariant provided the initial $\algg$-action
does not have classical anomaly. The standard interpretation if $S[\phi; A)$ is that it describes the coupling of $\phi$ to the background gauge field $A$ which can be then promoted to the dynamical field by adding the term responsible for the self-interaction of $A$.

Again, instead of considering $A$ as a background field, one can regard it as a dynamical one and modify the system in such a way that it does not bring additional local degrees of freedom. This is achieved by introducing the associated polymomentum field $B$, an $(n-2)$-form with values in $\algg^*$, and taking the extended action functional to be 
\begin{equation} \label{act+form-gauged}
    S[\phi, A, B] = S_0[\phi] + \int_X \tr \left(A (J + dB + \half[A,B])\right)\,.
\end{equation}
The equations of motion read as
\begin{equation}
    \begin{gathered}
        d A + \half [A,A] = 0\,, \\
        J + dB + [A,B] = 0\,, \\
        \delta S_0/ \delta \phi + A \delta J/ \delta \phi = 0\,,
    \end{gathered}
\end{equation}
and the gauge transformation for fields $A,B$ take the following form (modulo terms vanishing on-shell)
\begin{equation}
  \delta A = d\alpha +[A,\alpha]\,, \qquad \delta \phi=T_\alpha (\phi)\,, \qquad \delta B =d \mu+\commut{A}{\mu} + \commut{B}{\alpha}\,,
\end{equation}
where $T_\alpha(\phi)$ denotes the initial $\algg$-action on fields $\phi$ and the parameters are $\alpha \in \bigwedge^0(X,\algg)$, $\mu \in \bigwedge^{n-3}(X, \algg^*)$. If $X$ is simply connected one can employ the above gauge freedom to set $A=0$. If in addition, de Rham cohomology in degree $n-1$ is empty one can also solve the second equation for $B$ modulo its gauge equivalence. The residual gauge transformations in the sector of $\phi$ are those with $\alpha=const$ and coincide with the initial global transformations. At the same time the on-shell configurations of $B$ can not contribute to the space of solutions. In other words, the configurations related by the initial global symmetry transformations are now gauge equivalent and hence the space of gauge-inequivalent solutions is merely the quotient of the starting point one by the action of $G$. The above procedure is known in one or another version for long time, see e.g.~\cite{Rocek:1991ps} and can be employed to generate the dual formulation of the same degrees of freedom, see e.g.~\cite{delaOssa:1992vci,Quevedo:1997jb}. More recent discussion of gauging global symmetries and their generic degree analogues can be found in e.g.~\cite{Gaiotto:2014kfa,Bhardwaj:2023kri,Argurio:2024ewp}. 

In this work we employ presymplectic BV-AKSZ approach for gauging symmetries of generic degree in a systematic way. The approach is based on working in terms of the typically finite-dimensional $\fZ$-graded geometrical object, called presymplectic gauge PDE (presymplectic gPDE) underlying the local gauge theory in question, see~\cite{Grigoriev:2022zlq,Dneprov:2024cvt}, see also~\cite{Alkalaev:2013hta, Grigoriev:2016wmk,Grigoriev:2020xec,Dneprov:2022jyn} for earlier but less general versions. In this approach, a ghost variable and its associated (shifted) momentum are introduced for each symmetry generator. Under some technical conditions a graded Lie algebra of such symmetries gives rise to a new gauge field theory, the topological deformation of the initial system, which is constructed directly in the presymplectic BV-AKSZ formulation. In so doing the new ghost variables which are initially introduced as coordinates on the corresponding presymplectic gPDE give rise to new fields (unless the ghost degree is outside the range of possible degrees of spacetime differential forms) which can be seen as fields of the AKSZ sigma model for generalised BF-type topological theory, which is coupled to the starting point system. It is this topological model that is responsible for the change of global degrees of freedom. From a complementary perspective, the resulting theory is a consistent deformation, in the sense of~\cite{Barnich:1993vg}, of the direct sum of the starting point theory and a decoupled generalised BF-type theory, the deformation being the coupling between the two. This is what we refer to as a topological deformation. Note that such a deformation does not affect the presymplectic structure of the direct sum and hence belongs to the class of consistent deformations studied within the presymplectic BV-AKSZ approach in~\cite{Frias:2024vgw}. In this sense the present work explores a specific class of deformations, which, in contrast to those considered in~\cite{Frias:2024vgw}, involve an additional topological sector.

The gauging procedure is then applied to generate topological deformations of both types. We start with pairs of deformation type, obtained by gauging deformations, and consider a number of examples leading to the extended systems where the coupling constants associated to the deformations are promoted to dynamical fields whose equations of motion set them constants. Moreover, by gauging a positively-graded Lie algebras of symmetries one generally arrives at field theories where the fields corresponding to starting point coupling constants are also subject to algebraic constraints that ensure the consistency of the deformation. The instructive examples is the gauging of a graded Lie algebra controlling the general Yang--Mills type deformations of $N$-copies of Maxwell field. In so doing the associated degree zero ghosts are precisely the would-be structure constants
and it is not surprising that the corresponding Jacobi identity arises as the equation of motion for these fields on top of the constancy condition. Let us mention that the idea to promote coupling constants to dynamical/background fields whose equations of motion impose consistency of the theory was discussed already in~\cite{Horowitz:1986dta,Grigoriev:2006tt,deWit:2008ta} in various frameworks.

Pairs of duality type naturally arise by gauging nonpositively graded Lie algebras of symmetries. In this case the resulting theory typically has less solutions, because the degree $0$ symmetries act on solutions and their gauging effectively takes the quotient. The result of gauging is naturally interpreted as a dual formulation and generally has a different structure of gauge invariance as clearly illustrated by the example of gauging the shift symmetry of the massless scalar field in $4$-dimensions, resulting in the $2$-form gauge field. The relation between global symmetries and dualities is known for a long time, see e.g.~\cite{Quevedo:1997jb}. Our point is that this relation can be made into a systematic and preserving manifest locality construction of the dual formulation, which starts with a given gauge theory and the subalgebra of its nonpositively-graded symmetries.

Within the BV-BRST approach to local gauge theories, the present construction is related to the extended antifield formalism of~\cite{Brandt:1997cz}, proposed to control global symmetries and subsequently extended to consistent deformations and anomalies in~\cite{Barnich:1998qz,Barnich:2000me}, resulting in a powerful setup for perturbative renormalisation. An important difference with our approach is that the new ghost variables of the extended antifield formalism are constants rather than local fields and serve merely as a bookkeeping device for the $L_\infty$-algebra lift of a graded Lie algebra in local BRST cohomology. In particular, any variables conjugate to them are constants as well, while in our approach both the new ghosts and their conjugate momenta give rise to local form fields. For deformations the form fields corresponding to the momenta can be regarded as conjugate polymomenta and play the role of Lagrange multipliers, but for symmetries of nonpositive degree they become the gauge fields of the dual formulation. This being said, in this work we limit ourselves to Lie algebras of so-called internal symmetries, though the extension to $L_\infty$-algebras should be straightforward.

It is also worth mentioning a different perspective on systems resulting in the present approach from gauging deformations. Such gauge theories, e.g.\ unimodular gravity, can be obtained by imposing constraints on fields and gauge parameters so that the associated local BV-BRST complex is modified and typically involves reducible gauge symmetries~\cite{Kaparulin:2019quz,Abakumova:2020ajc}. If one reformulates the resulting gauge system in the presymplectic BV-AKSZ form, the higher degree ghosts associated to the reducibility relations give rise to the higher degree coordinates on the presymplectic gPDE and their associated higher-form fields.

The paper is organized as follows. In Section~\bref{sec:prelim} we recall the necessary background on gauge PDEs and presymplectic gauge PDEs. Section~\bref{sec:gauging} introduces symmetries of generic degree and the gauging procedure, first for gauge PDEs and then for presymplectic gauge PDEs, where it results in topological deformations. In Section~\bref{sec:deform-examples} we consider pairs of deformation type: the massive scalar, unimodular gravity, and Yang--Mills theory as a deformation of $N$ copies of Maxwell theory. Section~\bref{sec:dualities} is devoted to pairs of duality type, including the scalar--2-form, $p$-form, and principal chiral model--Freedman--Townsend dualities. A direct BV computation for the scalar--2-form duality and our sign conventions are collected in the Appendices.

\section{Preliminaries}\label{sec:prelim}

\subsection{Gauge PDEs} \label{sec:gPDE}

At the level of equations of motion, a concise and geometrical formulation of local gauge theories
is provided by the gauge PDE approach~\cite{Grigoriev:2019ojp,Grigoriev:2022zlq}. It combines the geometrical approach to PDEs~\cite{Vinogradov:1977,Vinogradov1981} with the jet-bundle version of the BV formalism~\cite{Barnich:1995ap,Barnich:1995db,Barnich:2000zw}.

Let us first recall the notion of a $Q$-manifold (also known as DG manifold). By definition a $Q$-manifold $(F,Q)$ is a $\mathbb{Z}$-graded manifold $F$ equipped with a homological vector field $Q$ that is a vector field of degree $1$ and is required to square to zero: $Q^2 = \half [Q,Q] = 0$. The $\fZ$-degree is denoted by $\gh{\cdot}$ and is referred to as the ghost degree. The algebra of functions on $F$ is naturally a differential graded commutative algebra, with the differential being $Q$. The standard example of a $Q$-manifold is the shifted tangent bundle $T[1]X$ with canonical $Q$-structure $\dx$ which is the de Rham differential of $X$ seen as a vector field on $T[1]X$. If $\{\theta^a, x^a\}$ are local coordinates on $T[1]X$ with ghost degrees $1$ and $0$ respectively, the vector field takes the form $\dx = \theta^a \ddl{}{x^a}$. $Q$-manifolds have a notion of equivalent reduction, namely, $(F,Q)$ is an equivalent reduction of $(F',Q')$ if $(F',Q') = (F,Q) \times (T[1]V, d_V)$ where $V$ is a graded vector space and $d_V$ is the canonical $Q$-field on the shifted tangent bundle to $V$.\footnote{More generally, one can allow $F'\to F$ to be a locally trivial $Q$-bundle with fibre $(T[1]V,d_V)$ admitting a global $Q$-section~\cite{Grigoriev:2019ojp}.} One also says that $(F_1,Q_1)$ and $(F_2, Q_2)$ are equivalent $Q$-manifolds if they are both equivalent reductions of some $(F,Q)$. Natural maps between $Q$-manifolds are maps between the manifolds which intertwine the differential, i.e.\ $\phi: (F_1,Q_1) \xrightarrow{} (F_2, Q_2) $ is a $Q$ map iff $Q_1\phi^* = \phi^* Q_2$. A $Q$-bundle is a bundle where both total space and the base space are $Q$-manifolds and the projection map is a $Q$-map~\cite{Kotov:2007nr}. When considering systems with fermions, $F$ carries an additional fermionic degree $\ferm{\cdot}$, and the supercommutative algebra structure is determined by the total parity $\p{f}=\gh{f}+\ferm{f}$ mod $2$, so that $fg=(-1)^{\p{f}\p{g}}gf$ for homogeneous $f,g\in\cC^\infty(F)$. For simplicity, in this work we restrict ourselves to systems without fermions, i.e.\ $\ferm{\cdot}=0$, so that $\p{f}=\gh{f}$ mod $2$.

A gauge PDE (gPDE) is a $Q$-bundle $(E,Q) \xrightarrow{\pi} (T[1]X,\dx)$, where $X$ is the spacetime manifold. Sections of the bundle are interpreted as fields, and equations of motion are encoded in the $Q$-structure: a solution is a section which is also a $Q$-section, i.e.\ for $\sigma: T[1]X \xrightarrow{} E$ the equations read:
\begin{equation}
    \dx \circ \sigma^* = \sigma^* \circ Q\,.
\end{equation}
Gauge transformations are parameterized by vector fields $Y$, $\pi_*Y = 0$ of ghost degree $-1$ called gauge parameters and are given by:
\begin{equation}
    \delta_Y \sigma^* = \sigma^* \circ [Q,Y]\,.
\end{equation}
Gauge-for-gauge transformations are encoded by vector fields $Z$, $\pi_*Z = 0$ of degree $-2$ and act on gauge parameters $Y$ by 
\begin{equation}
    \delta_Z Y = [Q,Z]\,.
\end{equation}
Gauge PDEs are considered equivalent if they are equivalent as $Q$-bundles over $T[1]X$.

The first natural example of a gauge PDE is constructed from a non-Lagrangian BV formulation on the infinite jet bundle, with the $Q$-structure given by the sum of the BV differential and the horizontal differential (seen as a vector field). All reasonable gauge PDEs arise as equivalent reductions of some gauge PDE corresponding to such a BV system on jets. One can show, that the space of solutions modulo gauge equivalences does not change under equivalent reductions of gauge PDEs. Another useful fact is that one can recover the non-Lagrangian BV formulation from the gPDE: the BV fields arise as components of supersections. Therefore gPDEs can be seen as a compact and flexible way to store the data of a local gauge field theory. For further details see \cite{Grigoriev:2022zlq,Grigoriev:2019ojp}.

\subsection{Presymplectic structures and Lagrangians}\label{sec:presymplect-gPDE}

To encode a Lagrangian, a gauge PDE is equipped with an additional structure, a compatible presymplectic structure~\cite{Alkalaev:2013hta,Grigoriev:2016wmk,Grigoriev:2022zlq}. A gauge PDE with compatible presymplectic structure is a gauge PDE $(E,Q)\xrightarrow{} (T[1]X,\dx)$ together with a two-form $\omega$ on the total space and a function $H$ of degree $n=\dim X$ obeying:
\begin{equation} \label{gpde-presymp-axioms}
    \begin{gathered}
        d\omega = 0\,, \qquad \gh{\omega} = n - 1\,, \\
        i_Q \omega + dH \in \cI\,,
    \end{gathered}
\end{equation}
where $\cI \subset \bigwedge^\bullet(E)$ is the differential ideal generated by the forms $\pi^*\alpha, \alpha \in \bigwedge^{1}(T[1]X)$. Our conventions for differential forms on graded manifolds are summarised in Appendix~\bref{app:conventions}.

It can be shown that these data define a BV formulation. Let us briefly recall how it is done. Let $x^a, \theta^a$ be coordinates on $T[1]X$ and $\psi^A$ are local fibre coordinates. The space of supersections $S\Gamma(E)$ is parameterized by superfields
\begin{equation}
\psi^A(x,\theta) = \ev^*(\psi^A)=\st{0}{\psi}{}^A(x) +\st{1}{\psi}{}^A_a (x)\theta^a + \frac{1}{2}\st{2}{\psi}{}^A_{ab} (x)\theta^a\theta^b + \dots\,,
\end{equation}
where $\ev:S\Gamma(E)\times T[1]X \to E$ is the corresponding evaluation map.
The structure $\omega$ defines a 2-form of ghost degree $-1$ on the space of supersections by:
\begin{equation} \label{AKSZ-BV-struc}
\Omega = \int_{T[1]X} \ev^*(\omega)=\int_{T[1]X}\delta \psi^A(x,\theta)\, \delta \psi^B(x,\theta)\,\omega_{AB}(\psi(x,\theta))\,,
\end{equation}
where we have introduced the components of $\omega$ via $\omega = d\psi^A d\psi^B\omega_{AB}$.
If this 2-form happens to be regular (i.e.\ has kernel of constant rank) one can take the quotient by the kernel distribution. The quotient is the symplectic space of the sought-for BV formulation. The BV action functional is given by 
\begin{multline} \label{AKSZ-BV-action}
    S_{BV} = \int_{T[1]X}(\ev^*{\chi})(\dx)+\ev^* (H)=\\
    \int d^nx\, d^n\theta \left( (\dx \psi^A(x,\theta))\,\chi_A(\psi(x,\theta)) + H(\psi(x,\theta))\right)\,,
\end{multline}
where $\chi = d\psi^A\chi_A$ is the presymplectic potential: $\omega = d\chi$. One can view this construction as a generalization of the AKSZ construction~\cite{Alexandrov:1995kv} for the case of non-topological field theories. The above BV action is the BV extension of the classical action which is conveniently defined as a functional on the space of sections $\sigma: T[1]X\to E$ and can be compactly written as
\begin{equation}
S[\sigma]=\int_{T[1]X}(\sigma^*{\chi})(\dx)+\sigma^* (H)\,.
\end{equation}
In what follows this action functional is referred to as the intrinsic action because it is defined in terms of the intrinsic geometry of the gPDE with presymplectic structure.

It is clear from the form of the intrinsic action that it is invariant under the algebraic gauge transformations generated by the regular vertical vector fields satisfying $i_V\omega\in \cI$. Vector fields satisfying this condition form an involutive distribution on $E$, which we call the kernel distribution of $\omega$. Moreover, not only the action but also the entire structure of the underlying gauge theory can be determined in terms of the quotient of $E$ by the maximal regular involutive subdistribution of the kernel distribution of $\omega$. Passing from the initial bundle $E$ to such a quotient is of course not an equivalence for the underlying gPDE, and it might even spoil the $Q^2 = 0$ condition. However, a condition which is preserved and actually ensures that the corresponding BV action satisfies the master equation is given by 
\begin{equation}\label{presymp-master}
    \half i_Q i_Q \omega + QH = 0\,.
\end{equation}
Together with conditions \eqref{gpde-presymp-axioms} this condition implies that $Q^2$ belongs to the kernel of the presymplectic structure. This motivates the following:
\begin{definition}\cite{Dneprov:2024cvt,Grigoriev:2022zlq}
A presymplectic gauge PDE is a bundle $E \xrightarrow{} T[1]X$ equipped with a vector field $Q$ of degree $1$, not necessarily nilpotent, that projects to $\dx$, a function $H$ of degree $n=\dim X$ and a presymplectic structure $\omega$, $\gh{\omega}=n-1$, obeying \eqref{gpde-presymp-axioms} together with \eqref{presymp-master}.
\end{definition} 
An alternative is to require $E$ to be equipped with $\chi$ and $H$ such that there exists $Q$ that projects to $\dx$ and the compatibility conditions are satisfied for $\omega\equiv d\chi$. In this version $\chi$ can be considered as an equivalence class modulo $\cI$ and these data define the corresponding (BV) action unambiguously. Indeed, the ambiguity in choosing $\chi$ such that $\omega=d\chi$ corresponds to the ambiguity in adding boundary terms to the intrinsic action.

Again, the data of a presymplectic gPDE defines a BV formulation by \eqref{AKSZ-BV-struc}, \eqref{AKSZ-BV-action}, provided the presymplectic structure $\ev^*(\omega)$ happens to be regular. The advantage of presymplectic gPDEs is that, in contrast to gPDEs, they can usually be taken finite-dimensional, if the underlying gauge system is local. Further details can be found in \cite{Grigoriev:2022zlq,Dneprov:2024cvt,Dneprov:2025eoi}.

\section{Gauging symmetries and deformations}\label{sec:gauging}

\subsection{Gauge PDE symmetries of arbitrary degree}

To begin with, let us consider the gauging of generic-degree symmetries of gauge systems defined at the level of equations of motion, i.e.\ when a Lagrangian/presymplectic structure is not specified or does not even exist. It turns out that in this case gauging is rather straightforward and has a clear interpretation.

Let us first recall the notion of a generic-degree symmetry for a gPDE. Let $(E,Q) \xrightarrow{\pi} T[1]X$ be a gPDE. A vertical symmetry is a vertical vector field $V$ on $E$ such that $[Q,V] = 0$ and $\gh{V}=k$. A $Q$-exact symmetry, i.e.\ a symmetry of the form $V= [Q,Y]$, is a gauge symmetry by definition of a gPDE. Global symmetries are equivalence classes of symmetries modulo gauge ones. The restriction to vertical vector fields is of a technical nature and does not lead to any loss of generality if one restricts to local analysis. More generally, one should consider projectable vector fields; see~\cite{Dneprov:2025eoi} for more details. Note that in the BV framework the notion of a nonvansihing-degree symmetry and the appropriate generalisation of the Noether theorem are known in the context of Lagrangian gauge theories, see e.g.~\cite{Barnich:2000zw} and references therein. In a more general BV-BRST setup for local gauge theories, symmetries of not necessarily vanishing degree were also discussed in~\cite{Sharapov:2016sgx}. The thorough discussion of such symmetries for (presymplectic) gPDEs can be found in~\cite{Dneprov:2025eoi}.

Symmetries of ghost degree $0$ act on sections in a natural way. Namely, if $\sigma: T[1]X \xrightarrow{} E$ is a section, then its infinitesimal variation determined by a vertical vector field $V$ reads as
\begin{equation}
\delta \sigma^* = \sigma^* \circ V\,.
\end{equation}
If $V$ is a symmetry, it is easy to check that the above variation is an infinitesimal symmetry of the equations of motion, i.e.\ if $\sigma$ is a solution of the underlying gPDE, then $\sigma + \delta \sigma$ is again a solution to first order in $V$. Symmetries of ghost degree $0$ form a module over the constants. It is more natural, however, to think of symmetries as a module over $\dx$-closed functions on the base. For degree $0$ symmetries, only degree $0$ functions are allowed, so that we recover a module over the constants if $X$ is connected.

Symmetries of degrees $-k < 0$ do not act on sections in a natural way. However, they admit a straightforward field-theoretical interpretation. First of all, when regarded as modules over constants, they describe level-$k$ reproducibilities of gauge (for gauge) parameters and hence are naturally regarded as symmetries of gauge (for gauge) parameters. In addition, if the de Rham cohomology of $X$ in positive degrees is nontrivial, negative-degree symmetries can give rise to the usual degree $0$ symmetries. Indeed, let $V$ be a symmetry of ghost degree $-k, k>0$. Then, for any $\alpha \in C^\infty(T[1]X)$, $gh(\alpha) = p, \dx \alpha = 0$, one finds that $\pi^*(\alpha)V$ is a symmetry of ghost degree $-k+p$. Moreover, in the case $\alpha = d_X \beta$, it follows that $\pi^*(\alpha)V = [Q, \pi^*(\beta) V]$, i.e.\ such a symmetry is gauge. This shows that equivalence classes of symmetries form a module over the spacetime de Rham cohomology.

The interpretation of symmetries of degree $1$ and higher 
is not so straightforward not to mention that just like the negative degree ones they do not act on sections. To unravel the meaning of positive degree symmetries, consider a symmetry $q$ of degree $1$ obeying $[q,q] = 0$. Such a symmetry encodes a possible deformation of the $Q$ structure, because $Q' = Q + q$ obeys $Q'^2 = 0$. Now more generally for a symmetry $q$, $\half [q,q] = q^2$ does not have to be zero. In this case $q$ encodes only an infinitesimal deformation of $Q$, i.e.\ a deformation $Q' = Q + tq$ with a formal parameter $t$ obeying $t^2 = 0$. Because $q$ is a symmetry, it follows that $q^2$ is always $Q$-closed, therefore in general $[q,q]$ gives a symmetry of degree $2$. The cohomological class of this symmetry encodes an obstruction to the second order deformation, a deformation $Q + tq + t^2(\dots)$ with formal parameter $t$ obeying $t^3 = 0$. Indeed, if $q^2$ is cohomologically trivial, i.e.\ $[q,q] = [Q,V]$ then $q$ determines a deformation up to cubic terms by $(Q + tq - \half t^2 V)^2 = 0 \quad mod(t^3)$. Let us mention that not all symmetries of degree $2$ and higher should come as obstructions to deformations of this sort. 

A more general setup is to consider a collection of symmetries which form a graded Lie algebra modulo the trivial ones that is modulo $Q$-exact vector fields. In particular, this happens if functions on $E$ form a module over $L_\infty$ algebra that acts via vertical fields. For the moment we refrain from considering such a more general situation.

\subsection{Gauging gPDE symmetries}

Now let us turn our attention to the gauging of symmetries. Usually, the term "gauging" refers to a procedure of coupling a given system to the gauge field associated to its global symmetry. In so doing a constant parameter of the global symmetry is promoted to an arbitrary spacetime function. In this work we concentrate on a very specific version of gauging where the gauge field is subject to a zero-curvature equation so that it does not bring new local degrees of freedom. It turns out that such a gauging can be implemented within the gPDE approach in a systematic and geometrical way.

Let $(E,Q)$ be a gPDE and $\algg$ a graded Lie algebra of its vertical symmetries, with basis $V_i$ obeying $[Q,V_i]=0$ and
\begin{equation}
    [V_i, V_j] = f^k{}_{ij} V_{k}\,.
\end{equation}
In particular, $\cC^\infty(E)$ is a $\algg$-module. 
The gauging of these symmetries results in the gPDE $(\tilde E,\tilde Q)$ with total space $\tilde{E} = E \times \algg[1]$ and $\tilde{Q} = Q + q_{CE}$, where $q_{CE}$ is the Chevalley--Eilenberg differential of $\algg$ with coefficients in $\cC^\infty(E)$. Besides the original system, $(\tilde E,\tilde Q)$ contains a topological sector of flat $\algg$-connections coupled to it through the term $c^iV_i$. It is therefore natural to refer to $(\tilde E,\tilde Q)$ as the topological deformation of $(E,Q)$ at the level of equations of motion. In contrast to the Lagrangian case considered in the next subsection, the conjugate momenta, and hence the $B$-fields of the BF-type sector, are absent. If $c^i$ denote coordinates on $\algg[1]$ then $\tilde Q$ can be explicitly written as:
\begin{equation}
\tilde Q=Q+q_{CE} = Q+c^i V_i + (-1)^{\p{j}} \frac{1}{2}f^k{}_{ij}c^j c^i \ddl{}{c^k}\,.
\end{equation}
Here and below we use the short notation $\p{i}\equiv \p{V_i}$. Recall that we assume $\ferm{\cdot}=0$, so that $\p{\cdot}=\gh{\cdot}$ mod $2$.

Introducing local coordinates $\psi^A$ parametrising the fibre of $E$ the space of sections is coordinatized by 
\begin{equation}
    \sigma^*c^i = A^i_{a_1...a_{k(i)}}(x) \theta^{a_1}...\theta^{a_{k(i)}} \,, \qquad \sigma^* \psi^A = \psi^A(x,\theta)\,,
\end{equation}
where $k(i) = \gh{c^i}$. Note that if $\gh{c^i}<0$ or $\gh{c^i}>\dim{X}$ then $\sigma^*(c^i)=0$ and hence the corresponding field does not arise. $A^i_{a_1...a_k}$ can be seen as a component of a $k$-form $A^i\equiv \sigma^*(c^i)$ on $X$. The equations of motion of $(\tilde E,\tilde Q)$ are given explicitly by
\begin{equation}
\begin{gathered}
    d_X A^i + \half[A,A]^i = 0\,, \\
    d_X\psi^A(x,\theta) = Q^A(\psi) + A^i V_i{}^A(\psi)\,.
\end{gathered}
\end{equation}

In the case when $\algg$ is an ordinary (degree 0) Lie algebra it is clear that equations of motion force $\algg$-valued 1-form fields $A^i \equiv \sigma^*(c^i)$ to satisfy zero curvature equation while the gPDE gauge transformation
is precisely the standard gauge transformation of $\algg$-connection. What happens to the degrees of freedom and the space of solutions after this operation? If the underlying spacetime manifold is connected and simply connected, then a flat connection is always a pure gauge. This means, one can choose a gauge with $A^i = 0$ and the resulting system has the equations of motion of the starting point gPDE $(E,Q)$. In this gauge, there are however residual gauge transformations which are precisely the global symmetries of the original theory. Therefore, the space of gauge-inequivalent solutions becomes smaller: it coincides with the quotient of the space of solutions of the original theory by the corresponding global symmetry. For a topologically nontrivial spacetime, additional global sectors arise from the moduli space of flat connections, while the original fields obey the twisted modification of the initial equations. Note that because $A^i$ is a globally defined 1-form, the underlying principal bundle will always be trivial as a bundle. If one wishes to account for non-trivial bundle structure, one has to perform the construction in each patch and then glue patches in a consistent way. This can be done systematically but goes beyond the scope of this paper.

Another natural extension concerns the case where $\algg$ is non-positively graded. Then the ghosts $c^i$ have degree $\geq1$, so that the fields $A^i$ are differential forms of positive degree, and the equations of motion state that $\sigma^*$ restricted to the ghost sector is a morphism of differential graded algebras $CE(\algg)\to\Omega(X)$, i.e.\ a flat higher connection. Such structures go back to Sullivan~\cite{Sullivan:1977pdi}, appeared in physics as free differential algebras~\cite{DAuria:1982uck} and are systematically employed in the unfolded formalism~\cite{Vasiliev:1988xc}. From the field-theory perspective, this corresponds to gauging higher-form symmetries~\cite{Gaiotto:2014kfa}, and the present construction provides an alternative framework for such gauging.

At the level of equations of motion, i.e.\ without a presymplectic structure, the gauging of degree-$0$ symmetries does not produce dual potentials, since the conjugate momenta are absent. For instance, gauging the shift symmetry of the massless scalar $\phi$ and eliminating the contractible pair $\phi$, $c$ leaves a formulation in terms of $\phi_a$ and its higher derivatives. This is closely related to the unfolded formalism, where dual formulations of a given system share the same gauge-invariant (Weyl) module and differ in the choice of gauge module, see e.g.~\cite{Matveev:2004ac, Skvortsov:2008vs,Skvortsov:2009nv,Boulanger:2008up}. The Lagrangian version of the construction, considered in the next subsection, supplements this with the conjugate momenta, which become the dual gauge fields.

Let us also comment on what happens if our gPDE is defined on the manifold with boundaries and the boundary conditions require the gauge parameters to vanish at the boundary. In this case the "gauge" transformations with constant parameters are no longer gauge and the gauged system is actually completely equivalent to the starting point one, provided the suitable boundary conditions are imposed on $A^i$ as well.

Now let us consider the case when $\algg$ is 1-dimensional and its basis element $V$ is of degree $1$. In particular, it follows that $V$ obeys $[V,V] = 0$. As discussed above, such a symmetry encodes a deformation of the $Q$-structure: $Q+\alpha V$ is nilpotent for any $\alpha\in\fR$ and hence defines a deformed gPDE structure on $E$. Applying our gauging procedure we extend the total space to $\tilde E = E \times \fR$. Parameterizing $\fR$ with coordinate $c$ of $\gh{c} = 0$ the space of sections of $\tilde{E}$ is parametrized by 
\begin{equation}
    \sigma^*c = \Lambda(x) \,, \qquad \sigma^* \psi^A = \psi^A(x,\theta)\,,
\end{equation}
where, as before, we denote by $\psi^A$ the fibre coordinate on $E$. The equations of motion of $(\tilde{E}, \tilde{Q})$
take the form
\begin{equation}
\begin{gathered}
    d_X \Lambda = 0\,, \\
    d_X \psi^A = \sigma^*(Q \psi^A) + \Lambda\sigma^*(V\psi^A)\,.
\end{gathered}
\end{equation}
These equations imply that $\Lambda$ is constant on any connected domain and hence can be regarded as a global conserved charge. Moreover, equations for $\sigma^*(\psi^A)$ are deformed compared to those for the initial gPDE $(E,Q)$, an extra vertex has been added, multiplied by the coupling constant $\Lambda$. This is the equations-of-motion counterpart of the deformation-type examples considered in Section~\bref{sec:deform-examples}.  

It is also instructive to consider a more general situation: suppose we gauge a graded algebra of symmetries concentrated in degrees $1$ and $2$ with $V_i\,, \gh{V_i} = 1$ obeying $[V_i, V_j] = f_{ij}^{\alpha} W_{\alpha}$ and $0 = [W_{\alpha}, W_{\beta}] = [V_{i}, W_{\alpha}]$. To gauge it we introduce coordinates $c^i$ and $c^\alpha_{-1}$ of ghost degrees $0$ and $-1$ correspondingly. The action of $\tilde Q$ is given by 
\begin{equation}
\label{12-rel}
\begin{gathered}
    \tilde Q c^\alpha_{-1} = -\half f^{\alpha}_{ij}c^i c^j\,, \\
    \tilde Q c^i = 0\,.
\end{gathered}
\end{equation}
The equations of motion in the sector of these variables read as
\begin{equation}
\dx \sigma^*(c^i)=0\,, \qquad f^\alpha_{ij}\sigma^*(c^i)\sigma^*(c^j)=0\,.
\end{equation}
The first equation implies that $\Lambda^i\equiv \sigma^*(c^i)$ are constant, while the second equation implies that the corresponding deformation of the initial system, namely $(E,Q_{\Lambda})$ with $Q_{\Lambda}\equiv Q+\Lambda^i V_i$, is consistent, i.e.\ $Q_\Lambda^2=0$. A Lagrangian counterpart of the same phenomenon is discussed in Section~\bref{sec:abelian2YM}, where the direct analogue of the above equations requires the structure constants parametrising the deformation to satisfy the Jacobi identity.

\subsection{Symmetries of presymplectic gPDEs and their gauging }

In the previous subsection we have seen that given a graded Lie algebra action compatible with $Q$ one can gauge it, resulting in a new gPDE whose space of inequivalent solutions differs, roughly speaking, by a finite-dimensional factor. More precisely, the modification depends on the space-time topology and, furthermore, can be much more involved if one considers systems where equations of motion have hidden differential consequences
or one is interested in finite values of global charges or finite symmetry transformations.

In this section we extend the gauging procedure to Lagrangian setting by employing the presymplectic BV-AKSZ approach. 
It turns out that the extension is rather straightforward if we restrict ourselves to symmetries that strictly preserve the presymplectic structure. In applications it is often more suitable to use presymplectic gPDEs rather than gPDEs with compatible presymplectic structures because in this way one can work with finite-dimensional spaces only. The construction of gauging arbitrary degree symmetries is immediately formulated in terms of presymplectic gPDE, bearing in mind that at any step it can be reformulated in terms of the corresponding gPDE with presymplectic structure. Details on the BV-AKSZ approach and presymplectic gauge PDEs are collected in Section~\bref{sec:presymplect-gPDE}, see also references therein.

Let $(E,Q,\omega)$ be a presymplectic gPDE. First, we need to recall what a symmetry of such an object is. For our purposes it is enough to limit ourselves to a subclass of symmetries, so called internal symmetries.
We say that a graded Lie algebra $\algg$ is represented on $E$ by internal symmetries, if the fundamental vector fields $V_i$ obey:
\begin{equation} \label{pres-top-def-cond}
\begin{gathered}
    \pi_* V_i = 0\,, \\
    [Q,V_i] \in \cK_\omega\,,\qquad [V_i, V_{j}] = f^k{}_{ij} V_k\,, \\
    L_{V_i}\omega \in \cI\,,
\end{gathered}
\end{equation}
where $\cK_\omega = \{X\in \mathfrak{X}(E): \pi_*X = 0, i_X \omega \in \cI \}$ is the vertical kernel distribution of $\omega$. If $\gh{i_{V_{i}}\omega}\neq 0$ then the last equation implies that there exist conserved currents $J_i\in \cC^\infty(E)$ satisfying 
\begin{equation}
\label{Vi-ham}
i_{V_{i}}\omega+dJ_i\in \cI\,.
\end{equation}
For $\gh{i_{V_{i}}\omega}=0$ such $J_i$ might exist only locally.
Furthermore, one can show that the internal symmetries of degree 0 define symmetries of the underlying intrinsic action and the charges associated to these symmetries are determined by $J_i$. It is clear that $J_i$ are determined by \eqref{Vi-ham}
modulo terms of the form $\pi^*f$, $f\in \cC^\infty(T[1]X)$. If de Rham cohomology of $X$ is trivial in degree $\gh{i_{V_i}\omega}+1$ then one can use the freedom in $J_i$ to achieve $QJ_i=0$. One can show that  
$(-1)^{|i|}V_{i} J_{j} = f^k{}_{ij} J_k + \pi^*(\kappa_{ij})$, where $\kappa_{ij}$
are coefficients of the $2$-cocycle of $\algg$, which are also functions on $T[1]X$. If $\kappa_{ij}$ happen to be a trivial cocycle one can show that it can be removed by a redefinition of $J_i$. In general, $\kappa_{ij}$ is the classical anomaly arising from the fact that Noether charges represent the symmetry algebra projectively. Further details on symmetries of presymplectic gauge PDEs can be found in~\cite{Dneprov:2025eoi}.

\begin{prop}\label{prop:main}
Let $(E,Q,\omega)$ be a presymplectic gPDE, graded Lie algebra $\algg$ acts on $E$ by internal symmetries, and the corresponding currents $J_i$ be defined globally and satisfy $QJ_i=0$. If in addition the classical anomaly vanishes, i.e.\ $(-1)^{|i|}V_{i} J_{j} = f^k{}_{ij} J_k$, then there exists a new presymplectic gPDE $E' = E \times T^*[n-1]\algg[1]$ over $T[1]X$ whose $Q$-structure is given by
\begin{multline}
    Q' = Q + c^i V_i + (-1)^{|j|}\frac{1}{2}f^k{}_{ij} c^j c^i \frac{\partial}{\partial c^{k}} + 
    \\
    ((-1)^{|i|(n+1)}J_i -
    (-1)^{|j| + |i|+ n} f^k{}_{ij} b_k c^j )\frac{\partial}{\partial b_{i}}\,,
\end{multline}
and 
\begin{equation}
    \omega' = \omega + \omega_{can}\,,\qquad \omega_{can}=db_idc^i\,,
\end{equation}
where $c^i$, $\gh{c^i}=1-\gh{V_i}$ are linear coordinates on $\algg[1]$ and $b_i$, $\gh{b_i}=n-1-\gh{c^i}$ are the corresponding coordinates on the fibres of $T^*[n-1]\algg[1]$.
\end{prop}
\begin{proof}
One can explicitly check that the axioms of a presymplectic gPDE hold for the deformed structures. Indeed,
\begin{equation}
    i_{Q'} \omega + d(H) + c^i d(J_i) \in \cI\,,
\end{equation}
together with 
\begin{equation}
    i_{Q'} \omega_{can} = (-1)^{|i|}dc^iJ_i + d(-\half (-1)^{(|j|+1)} f^k{}_{ij} b_k c^j c^i )\,,
\end{equation}
gives   
\begin{equation}
    i_{Q'} \omega' + d(H -(-1)^{|i|} c^i J_i + \half (-1)^{(|j|+1)} f^k{}_{ij} b_k c^j c^i ) \in \cI\,,
\end{equation}
i.e.\ $Q'$ is Hamiltonian with $H'\equiv H -(-1)^{|i|} c^i J_i + \half (-1)^{(|j|+1)} f^k{}_{ij} b_k c^j c^i$. Finally, 
\begin{multline}
    \half i_{Q'} i_{Q'} \omega' + Q'H' = \half i_Q i_Q \omega + QH + \\ + i_{c^iV_i}(i_Q\omega + dH) + \\ + (\half i_{c^iV_i}i_{c^iV_i}\omega + c^j c^i V_i J_j - (-1)^{|i|} \frac{1}{2}f^k{}_{ij} c^j c^i J_k ) = 0\,.
\end{multline}

\end{proof}
In what follows we refer to $(E',Q',\omega')$ as the topological deformation of $(E,Q,\omega)$ obtained by gauging the symmetries $V_i$.

Let us now discuss the resulting system in some details. If $\psi^A$ are coordinates on the fibres of $E$ then the space of sections of $E'$ can be coordinatized by
\begin{equation}
\begin{gathered}
    \sigma^*(\psi^A) = \psi^A(x,\theta)\,, \\
    \sigma^*c^i = A^i(x,\theta)\,, \qquad \sigma^* b_i = B_i(x,\theta)\,.
\end{gathered}
\end{equation}
The intrinsic action for $(E',Q', \omega')$ reads as
\begin{equation}
    S' = S[\psi] + \int_{T[1]X} \, -(-1)^{|i|}A^i J_i + B_k(\dx A^k + (-1)^{(|j|+1)}\half f^k{}_{ij} A^j A^i)\,,
\end{equation}
where $S[\psi^A]$ stands for the intrinsic action for $(E,Q,\omega)$ and 
by some abuse of notations $J_i$ denotes $\sigma^*(J_i)$. It is clear that the modification amounts to coupling the initial gauge theory to the BF-type topological system for fields $A,B$, with the coupling being $(-1)^{|i|+1}A^iJ_i$. In this sense one can indeed regard $(E',Q',\omega')$ as a topological deformation of the initial gauge system.

It is also instructive to inspect the equations of motion determined by $S'$:
\begin{equation} \label{eom-deform-top}
    \begin{gathered}
        E_A - (-1)^{|i|+|\psi^A|+|i||\psi^A|}A^i \frac{\delta J_i}{\delta \psi^A} = 0\,, \\
        d_X A^i + (-1)^{(|k|+1)}\half f^i{}_{jk} A^k A^j = 0\,, \\
        d_X B_i + (-1)^{n+|i|+|j|} f^k{}_{ij} B_k A^j -(-1)^{|i|(n+1)} J_i = 0\,,
    \end{gathered}
\end{equation}
where $E_A$ denote the Euler--Lagrange equations of $S[\psi]$. The gauge transformations for the fields $A^i, B_j$ are determined by $Q'$ and read as:
\begin{equation}
\begin{gathered}
    \delta A^i = d_X \alpha^i - f^i{}_{jk} A^k \alpha^j\,, \\
    \delta {B}_i = d_X \mu_i - (-1)^{n + |i| + |j|} f^k{}_{ij} \mu_k A^j - f^k{}_{ij} B_k \alpha^j\,,
\end{gathered}
\end{equation}
where $\alpha$ and $\mu$ are gauge parameters that are arbitrary functions on $T[1]X$ of degree $-|i|$ and $n+|i|-3$ respectively. Note that any such function can be seen as a differential form on $X$ and it identically vanishes if its degree is negative or greater than $n=\dim X$. The first two equations in \eqref{eom-deform-top} are independent of $B_i$. The second equation in \eqref{eom-deform-top} is a Maurer--Cartan equation for the $\algg$-valued (super)connection $A^i$.

Let us analyse the structure of solution to the above system of equations assuming that the base space $X$ is contractible. We first consider the case where $n>\gh{c^i}>0$ and $\gh{b_i}\geq 0$. In this case the second equation defines a flat connection or more precisely $\fZ$-graded version of Quillen superconnection~\cite{Quillen:1985vya}. Under our assumptions it always has a solution, e.g.\ $A^i=0$, and all solutions are gauge equivalent. In other words $A$ does not contribute to the moduli of solutions unless $X$ is non contractible. In its turn, the first equation of \eqref{eom-deform-top} is the equations of motion for physical fields modified by the coupling to $A^i$. Finally, on solutions the currents are covariantly closed, $D_A\sigma^*(J)=0$, where $D_A=d_X+\rho_A$ and $\rho_A$ denotes the coadjoint action of the $\algg$-valued form $A$. Indeed, the equations of motion imply $d_X\circ\sigma^*-\sigma^*\circ Q'=\sigma^*\circ k$ for some vertical $k\in\cK_\omega$. Since $kJ_i=0$ for any $k\in\cK_\omega$, as follows from \eqref{Vi-ham}, $QJ_i=0$ and the currents are equivariant, i.e.\ $(-1)^{|j|}V_jJ_i=f^k{}_{ji}J_k$, one finds $d_X\sigma^*(J_i)=\sigma^*(c^jV_jJ_i)$, which is precisely $D_A\sigma^*(J)=0$. One can show that on a contractible manifold under our conditions this implies that the third equation can be solved for $B_i$, see e.g. \cite{Arias_Abad_2019}. In the critical case $\gh{c^i} = n$, one has $\gh{J_i} = 0$ and the last equation could impose restrictions on possible values of $J_i$. 

Note that under some technical conditions, this system with $n=1$ and $\gh{J_i}=0$ describes a first-class constrained system with constraints $J_i$ or equivalently implements the Marsden--Weinstein reduction to the zero value of the moment map $E \to \algg^*$ whose coefficients are $J_i$. In both cases the moduli of solutions is roughly speaking the gauge quotient of the surface $J_i=0$.  

Another typical situation corresponds to gauging deformations and in this case $\gh{c^i}=0$ and hence $\gh{b_i}=n-1$. The second equation in \eqref{eom-deform-top} then implies $A^i=const$ and the first equation is the deformation of the starting point $E_A$ where $A^i$ are naturally interpreted as coupling constants. For $n>1$ the third equation neither brings obstructions nor contribute to the moduli of solutions. An interesting phenomenon that occurs in a more general case $\gh{c^i}\leq 0$ is that the second equation generally imposes algebraic constraints on the values of $A^i$ of vanishing degree.

To conclude, just like the non-Lagrangian case the sector of $A^i,B_i$ is topological in the sense that it does not affect local degrees of freedom but can change the moduli of solutions by a finite-dimensional factor in the way sensitive to the global spacetime geometry. In the next sections we present a few examples of field theories which are related by the topological deformations of the type described in this section.

\section{Examples of gauging deformations}\label{sec:deform-examples}

In this section we consider examples of pairs of deformation type. In the simplest case the gauged symmetry is of degree $1$, so that the corresponding ghost has degree $0$ and its momentum has degree $n-1$. The field associated with the ghost is a dynamical coupling constant, which is constant on-shell, while the $(n-1)$-form associated with its momentum enters as a Lagrange multiplier. The value of the coupling can be recovered as the on-shell value of the charge associated with the degree $-(n-1)$ symmetry $\partial/\partial b$. In the Yang--Mills example of Section~\bref{sec:abelian2YM} we also consider a much more general gauging, where the gauged symmetries form a graded Lie algebra of degrees $1$ and higher.

\subsection{Massive scalar}\label{sec:massive-scalar}

As a first illustration of gauging a deformation we consider the toy example of massless field and its massive deformation. The presymplectic gPDE describing the massless scalar is a trivial fibre bundle $E = T[1]X \times W$ over $T[1]X$, where the fibre $W$ is coordinatized by real $\phi, \phi_a$ of ghost degree $0$. The $Q$-structure is given by
\begin{equation}
\begin{gathered}
    Q \phi =\theta^{a} \phi_{a}\,, \qquad Q \phi_{a} = 0\,,
\end{gathered}
\end{equation}
and the presymplectic structure reads as
\begin{equation}
    \omega = (-1)^{(n-1)}d(\theta^{(n-1)}_a \phi^a d\phi) \,, \qquad n = \dim(X)\,,
\end{equation}
where we assume that we work in the mostly-negative signature. Here and in what follows we employ the following short-hand notation:
\begin{align}
    \theta^{(k)}_{a_1\dots a_{n-k}}\equiv\frac{1}{k!}\epsilon_{a_1\dots a_{n-k}b_1\dots b_{k}}\theta^{b_1}\dots\theta^{b_{k}}\,,
\end{align}
where $\epsilon_{0\,1\ldots \,n-1}=\epsilon^{{0\,1\ldots \,n-1}}=1$. Note that $\epsilon^{a_1\ldots a_n}$ is not obtained from $\epsilon_{a_1\ldots a_n}$ by raising indices with the metric. The Berezin integration over $\theta$ is normalised by $\int d^n\theta\, \theta^{(n)}=1$.

This presymplectic gPDE possesses a degree-$1$ symmetry 
\begin{equation}
\begin{gathered}
    V_M = -\frac{1}{n}\theta^a \phi \frac{\partial}{\partial \phi^a}\,,
\end{gathered}
\end{equation}
which obeys
\begin{equation}
[Q,V_M] \in \cK_\omega, \qquad i_{V_{M}}\omega + dJ_M \in \cI\,,
\end{equation}
with 
\begin{equation}
J_M = -\frac{1}{2}\theta^{(n)} \phi^2\,.
\end{equation}

Following the general prescription of Proposition~\bref{prop:main}, we take $E'=E\times \fR\times \fR[n-1]$ with extra coordinates $c, b$ of degrees $0, n-1$ and take the following presymplectic structure:
\begin{equation}
    \omega' = (-1)^{(n-1)}d(\theta^{(n-1)}_a \phi^a d\phi) + db\, dc\,.
\end{equation}
The $Q$ structure of the deformed presymplectic gPDE is given by:
\begin{equation}
\begin{gathered}
    Q' \phi = \theta^a \phi_a\,, \\
    Q' \phi_a = -\frac{1}{n}\theta_a c \phi\,, \\
    Q' b = \frac{(-1)^n}{2}\theta^{(n)} \phi^2 = (-1)^{(n+1)}J_{M}, \, \qquad Q'c = 0\,,
\end{gathered}
\end{equation}
and is a Hamiltonian vector field with respect to $\omega'$ with the Hamiltonian being 
\begin{equation}
    H' = -\theta^{(n)} \frac{1}{2} c \phi^2 - \frac{1}{2} \theta^{(n)} \phi^a \phi_a\,.
\end{equation}
Proposition~\bref{prop:main} implies that $(E',Q',\omega')$ is a presymplectic gauge PDE which is a topological deformation of the initial one.

Let us see what the deformed system is. Parametrising sections by
\begin{equation}
    \begin{gathered}
        \sigma^* \phi = \phi(x) \, , \qquad \sigma^* \phi^a = \phi^a(x)\,, \\
        \sigma^* b = B_{b_1...b_{n-1}}(x)\theta^{b_1} ... \theta^{b_{n-1}}\,, \qquad \sigma^* c = M(x)\,,
    \end{gathered}
\end{equation}
the intrinsic action reads as
\begin{equation} \label{action-deformed-scalar}
    S' = \int d^nx \ \phi^a \partial_a\phi - \frac{1}{2} \phi^a \phi_a - M(\frac{1}{2}\phi^2 - (-1)^n \epsilon^{a_1...a_n}\partial_{a_1} B_{a_2...a_n})\,,
\end{equation}
or, after the elimination of $\phi_a$ through its equations of motion $\phi_a = \partial_a \phi$:
\begin{equation} \label{action-deformed-scalar-2nd}
    S' = \int d^nx \ \frac{1}{2}\partial^a\phi \partial_a\phi - \frac{M}{2}\phi^2 + (-1)^n \epsilon^{a_1...a_n} M\partial_{a_1} B_{a_2...a_n}\,.
\end{equation}
Its equations of motion are given by:
\begin{equation}\label{equations-deformed-scalar}
\begin{gathered}
    (\Box + M)\phi = 0\,, \\
    \phi^2 = 2(-1)^n\epsilon^{a_1...a_n}\partial_{a_1} B_{a_2...a_n}\,, \\
    \partial_a M= 0\,.
\end{gathered}
\end{equation}

We see that in this example the field $M$ becomes a constant on the equations of motion which for $M>0$ plays the role of the square of the mass parameter. It can be recovered as a conservation law with respect to the ghost degree $-(n-1)$ symmetry $\ddl{}{b}$.

Let us finally comment on what happens if we repeat the above considerations for the Laplace equation on a compact space, e.g.\ $S^2$. The generalisation is straightforward but it turns out that the curved generalisation of the second equation in \eqref{equations-deformed-scalar} has the form
$dB \propto \phi^2$ and hence does not have global solutions unless $\phi=0$. An analogous phenomenon occurs in the global versions of the next example.

\subsection{Unimodular gravity as a topological deformation}

In this subsection we obtain the presymplectic gPDE formulation of the unimodular gravity as a topological deformation of Einstein gravity. 

We start from the formulation of gravity as presymplectic gPDE, proposed in \cite{Grigoriev:2020xec}. The corresponding $Q$-bundle is globally trivial and is given by $(E,Q)=(\mathfrak{iso}(3,1)[1],d_{CE})\times (T[1]X,\dx)$, where $\mathfrak{iso}(3,1)$ is the Poincar\'e algebra in $4$ dimensions and $d_{CE}$ its differential. For simplicity we restrict ourselves to the case $\dim{X}=4$ but the extension to generic $n>2$ is straightforward.\footnote{For simplicity we also assume that $X$ is parallelisable; otherwise $E$ has this product form only locally.} Introducing coordinates as
\begin{center}
\begin{tabular}{| c | c | c |} 
 \hline
 coordinate & $\xi^a$\,, $\rho^{ab}$\,, $\theta^\mu$ & $x^\mu$ \\ 
 \hline
 ghost & 1 & 0 \\ 
 \hline
\end{tabular}
\end{center}
where $\xi^a$ correspond to translations and $\rho^{ab}$ to Lorentz generators of $\mathfrak{iso}(3,1)$. $Q$-vector is determined by 
\begin{equation}
\begin{gathered}
     \quad Q \xi^a = -\rho^a{}_b \xi^b, \quad Q\rho^{ab} = -\rho^{ac} \rho_c{}^b \,, \quad Qx^\mu=\theta^\mu\,,
\end{gathered}
\end{equation}
and the presymplectic structure is given by~\cite{Alkalaev:2013hta} 
\begin{equation}\label{presymplect_unimod}
    \omega = d\chi= -d \left(\frac{1}{32\pi}\epsilon_{abcd} \xi^a \xi^b d\rho^{cd} \right)\,.
\end{equation}
It is easy to check that the compatibility conditions are satisfied so that we have a presymplectic gauge PDE and that it encodes the proper BV formulation of Einstein gravity.

The above presymplectic gPDE possesses the following degree $1$ symmetry:
\begin{equation}
\begin{gathered}
    Y \rho^{ab}= - \frac{1}{3}\xi^a \xi^b\,,
\end{gathered}
\end{equation}
obeying $[Q,Y] = 0$, $[Y,Y] = 0$. Its associated conserved current is given by:
\begin{equation}
    J_Y = -\frac{1}{8\pi}\xi^{(4)}\,, \qquad 
    i_Y \omega + dJ_Y = 0\,.
\end{equation}
Applying the gauging procedure of Proposition~\bref{prop:main} we extend $E$ by coordinates $c$ and $b$ of degree $0$ and $3$ respectively. The action of differential $Q'$ is given by:
\begin{equation}
\begin{gathered}
    Q'b = \frac{1}{8\pi}\xi^{(4)} = -J_Y\,, \\
    Q'c = 0\,,
\end{gathered}
\end{equation}
\begin{equation}
\begin{gathered}
    Q' \xi^a = -\rho^a{}_b \xi^b\,, \\
    Q'\rho^{ab} = -\rho^{ac} \rho_c{}^b - \frac{1}{3} c\, \xi^a \xi^b\,,
\end{gathered}
\end{equation}
and the extended presymplectic structure reads as:
\begin{equation}\label{presymplect_unimod-deformed}
    \omega' = -\frac{1}{32\pi} d(\epsilon_{abcd} \xi^a \xi^b d\rho^{cd}) +db\, dc\,.
\end{equation}
Parametrising the space of sections by

\begin{equation}
\begin{gathered}
    \sigma^*(\xi^a) = e^a_{\mu} \theta^\mu\,, \\
    \sigma^*(\rho^{ab}) = \omega^{ab}_{\mu} \theta^\mu\,, \\ \sigma^*(c) = \Lambda\,, \qquad \sigma^*(b) = \frac{1}{3!}B_{\mu\nu\rho}\theta^\mu \theta^\nu\theta^\rho\,,
\end{gathered}
\end{equation}
the associated intrinsic action takes the form:
\begin{equation}
    S' = \int_X -\frac{1}{16\pi}e^{(2)}_{ab}(d \omega^{ab} +\omega^{a}_c \omega^{cb}) - \frac{1}{8\pi} \Lambda(e^{(4)}-8\pi d B)\,.
\end{equation}
This action is a frame-like version of the action proposed by Henneaux in~\cite{Henneaux:1989zc}. Let us also mention somewhat related analysis of unimodular gravity from~\cite{Kaparulin:2019quz}.

Equations of motion are:
\begin{equation}
\begin{gathered}
    \frac{\delta S'}{\delta \omega} = 0 \leftrightarrow de^a + \omega^{ab}e_b = 0\,, \\
    \frac{\delta S'}{\delta e} = 0 \leftrightarrow \epsilon_{abcd}e^b(R^{cd} + \frac{\Lambda}{3}e^ce^d) = 0\,, \\
    \frac{\delta S'}{\delta B} = 0 \leftrightarrow d \Lambda = 0\,, \\
    \frac{\delta S'}{\delta \Lambda} = 0 \leftrightarrow \ e^{(4)} - 8\pi dB =0\,,
\end{gathered}
\end{equation}
where $e_\mu^a e_a^\nu=\delta_\mu^\nu$ and 
\begin{equation}
\begin{gathered}
    R^{ab} = d\omega^{ab} +\omega^{ac} \omega_c{}^b = \frac{1}{2}R^{ab}_{|\mu\nu} \theta^\mu \theta^\nu\,.
\end{gathered}
\end{equation}
The first equation can be solved for $\omega^{ab}_\mu = \omega^{ab}_\mu(e)$ stating that it is the Levi-Civita spin connection. The third equation requires $\Lambda$ to be constant. The second equation is then just the Einstein equation written in terms of frame variables with cosmological constant $\Lambda$. The last equation can be locally solved for $B$ up to its gauge freedom with no restrictions on the gravitational variables. Globally however, if $B$ is understood as a globally defined form, this equation states that the volume form is exact, which may restrict the space of solutions. 

Finally, note that $\Lambda$ is an on-shell conserved quantity that can be recovered as a conserved current associated to a ghost degree $-3$ symmetry $\ddl{}{b}$. This was observed in~\cite{Elfimov:2021aok} within a related framework.

\subsection{Yang--Mills theory as a topological deformation}\label{sec:abelian2YM}

Our next example is the classic consistent deformation of $N$ copies of Maxwell theory into the nonabelian Yang--Mills theory, see~\cite{Barnich:1993vg,Barnich:1993pa,Barnich:1994mt} for the local BRST cohomology approach. We use the minimal presymplectic gPDE formulation of Maxwell theory~\cite{Grigoriev:2022zlq}, see also~\cite{Dneprov:2022jyn,Dneprov:2024cvt}, with $N$ copies of ghosts denoted by $C^i$ and $N$ copies of the Faraday tensor, denoted by $F^i_{ab}$, $i=1,\ldots N$. The underlying internal space is isomorphic to $\fR^N$ and denoted by $W$ so that the fibre of this presymplectic gPDE is given by $W[1]\times (W \tensor (\fR^4 \wedge \fR^4))$. The $Q$-structure and the presymplectic structure are given explicitly by:
\begin{equation}
\begin{gathered}
    QC^i = \frac{1}{2}F_{ab}^i \theta^a \theta^b\,, \qquad QF^i_{ab}=0\,, \\
    \omega = d\chi\,, \qquad \chi =-\frac{1}{2}\theta^{(2)}_{ab}F^{i,ab}dC^j \delta_{ij}\,.
\end{gathered}
\end{equation}
In what follows we simply take $\delta_{ij}$ to be a unit matrix.

The consistent deformations for this system in the presymplectic BV-AKSZ approach were discussed in~\cite{Frias:2024vgw} to which we refer for a detailed discussion. Now we only consider the following concrete deformation leading to YM theory:
\begin{equation}
\begin{gathered}
    VC^k = -\frac{1}{2}f^k{}_{ij}C^iC^j\,, \\
    V F_{ab}^k = f^k{}_{ij} F_{ab}^{[i}C^{j]}\,,
\end{gathered}
\end{equation} 
where $f^k{}_{ij}$ are components of a totally antisymmetric constant tensor so that $V$ is Hamiltonian with respect to the above presymplectic structure. The additional conditions 
\begin{equation}
[Q,V]\in \cK_\omega\,, \qquad [V,V] = 0\,,
\end{equation}
ensure that $V$ realises the action of a symmetry algebra (in this case a 1-dimensional degree $1$ Lie algebra) via the vector field $V$. These are obviously satisfied if the map $W\wedge W \to W$ determined by $f^k{}_{ij}$ is a Lie algebra structure on $W$. Note that total antisymmetry of $f^k{}_{ij}$ means that the Euclidean inner product on $\fR^N$ is invariant. The current associated to $V$ reads explicitly as $J_V=-\frac14 \theta^{(2)}_{ab}F^{l,ab} f^k{}_{ij}C^iC^j\delta_{lk}$ 

It is straightforward to apply the gauging procedure to the above symmetry, giving an extended presymplectic gPDE with additional coordinates $\gamma,\gh{\gamma}=0$ and $\beta, \gh{\beta}=3$. The meaning of the field $\sigma^*(\gamma)$ is obvious -- this is a dynamical coupling constant which is subject to $d \sigma^*(\gamma)=0$. Let us also present explicitly the expression for the intrinsic action. To this end introduce the following parametrization of sections:
\begin{equation}
\sigma^*(C^i)=A^i \,, \qquad \sigma^*(F^i_{ab})=F^i_{ab}\,, \qquad 
\sigma^*(\gamma)=\lambda\,, \qquad \sigma^*(\beta)=B\,,
\end{equation}
where $\lambda$ is a $0$-form and $B$ is a $3$-form. The intrinsic action then reads as: 
\begin{equation}
\label{deformedYM}
S'=\int d^4x\Big(-\half F^{ab}_i\big(\partial_aA^i_b-\partial_bA^i_a
+\lambda f^i{}_{jk}A^j_aA^k_b\big)+\frac14 F^{ab}_iF_{ab}^i
-\lambda\,\partial_a B^a\Big)\,,
\end{equation}
where $B^a=-\frac{1}{3!}\epsilon^{abcd}B_{bcd}$. Varying $B^a$ gives
$\partial_a\lambda=0$, while for constant $\lambda$ the remaining equations for $A^i,F^i$ are those of Yang--Mills theory with coupling constant $\lambda$. Varying $\lambda$ gives $\sigma^*(J)+dB=0$ which can be solved with respect to $B$, at least locally.

In fact this example can be employed to illustrate a much more general class of gaugings. To see this let us recall that the inner product $\delta_{ij}$
on $W$ defines a symplectic structure of degree $2$ on $W[1]$ and 
consider the Lie algebra $\algg$ of Hamiltonian vector fields on $W[1]$ of homogeneity degree $1$ and higher. It naturally acts on $W[1]$ seen as a graded manifold and, moreover, its action extends to $W[1]\times W$, which in turn naturally extends to $W[1]\times (W \tensor (\fR^4 \wedge \fR^4))$.
Indeed, the action on $W[1]$ is canonical and its shifted tangent lift to $T[-1](W[1])$ is canonical as well while $W[1]\times W$ can be identified with $T[-1](W[1])$. Because this action is fibrewise linear, it extends to the action on $W[1]\times (W \tensor (\fR^4 \wedge \fR^4))$. In plain terms this means that the action does not affect spacetime indices of $F^i_{ab}$. By extending this action to the vertical action on $E$ one ends up with the $\algg$-action on $E$. 

In order to see that $\algg$-action on $E$ is compatible with $Q$ and $\omega$ let us first give useful explicit formulas for it. To this end it is convenient to introduce the following vector field $X_{ab}=F^i_{ab}\dl{C^i}$ so that
$Q=\dx+\theta^{ab}X_{ab}$, where $\theta^{ab}=\half \theta^a \theta^b$. The canonical lift $\bar u$ of $u\in \algg$ defined on $W[1]$ to $W[1]\times (W \tensor (\fR^4 \wedge \fR^4))$ can be written as follows
\begin{equation}
\bar u C^i\equiv uC^i\,, \qquad \bar u F^i_{ab} \equiv (-1)^{\p{u}}X_{ab}uC^i\,.
\end{equation}
It follows $[X_{ab},\bar u]C^i=0$ and hence $[Q,\bar u]C^i=0$ as well. Then it is easy to check that $i_{[Q,\bar u]}\omega=0$ so that $[Q,\bar u]\in \cK_{\omega}$.

To see that $\bar u$ is Hamiltonian let us observe that the presymplectic structure can be written as 
\begin{equation}
\omega = \half \theta^{(2)}_{ab} dF^{i,ab} dC^j \delta_{ij}=
-\frac{1}{4} \left(\theta^{(2)}_{ab} L_{X^{ab}} 
(dC^idC^j\delta_{ij})\right)\,,
\end{equation}
where all the equalities are understood modulo terms proportional to $d\theta^a$, i.e.\ terms from $\cI$. It follows that
\begin{multline}
L_{\bar u}\omega=-\frac{1}{4} \theta^{(2)}_{ab} L_{\bar u} L_{X^{ab}} \left( dC^idC^j\delta_{ij} \right)=
-(-1)^{\p{u}}\frac{1}{4}\theta^{(2)}_{ab}L_{X^{ab}} L_{\bar u} (dC^idC^j\delta_{ij})=\\
-(-1)^{\p{u}}\frac{1}{4}\theta^{(2)}_{ab} L_{X^{ab}} L_{u} (dC^idC^j\delta_{ij})=0\,,
\end{multline}
where the last equality follows from the condition that $u$ is Hamiltonian with respect to $dC^idC^j\delta_{ij}$.  

The remaining condition of Proposition~\bref{prop:main} is that the algebra of Hamiltonians for the $\algg$-action on $E$ with respect to $\omega$ is not centrally extended. First of all, the algebra of Hamiltonians of $\algg$-action on $W[1]$ with respect to presymplectic structure $dC^idC^j \delta_{ij}$ is not centrally extended if one assumes Hamiltonians to be $x,\theta$-independent. Furthermore, if $H_u$ is a Hamiltonian of $u\in \algg$ with respect to $dC^idC^j \delta_{ij}$ then the Hamiltonian of its lift $\bar u$ can be written as $\frac14\theta^{(2)}_{ab} {X^{ab}}H_u$ which in turn implies that
\begin{equation}
\bar u \theta^{(2)}_{ab} {X^{ab}}H_v=(-1)^{\p{u}} \theta^{(2)}_{ab} {X^{ab}} uH_v= (-1)^{\p{u}} \theta^{(2)}_{ab} {X^{ab}}\pb{H_u}{H_v}\,,
\end{equation}
so that the classical anomaly vanishes. Following the same argument, one can easily see that $QJ_{\bar u} = 0$ if we choose Hamiltonians to be independent of $\theta, x$.

Let $V_{k i_1\ldots i_l}\equiv C_{[i_1}\ldots C_{i_l}\dl{C^{k]}}$, $2\leq l\leq N-1$, be a basis of $\algg$, where here and below the indices of $W$ are raised and lowered with $\delta_{ij}$ and we do not distinguish their position. Note that $V_{k i_1\ldots i_l}$ are totally antisymmetric, reflecting the fact that $\algg$ is isomorphic (up to a degree shift) to the Lie algebra of functions on $W[1]$ of homogeneity at least
$3$, with the bracket given by the Poisson bracket.
Our gauging procedure amounts to introducing ghosts $\gamma^{k i_1\ldots i_l}$ and their conjugate momenta $\beta_{k i_1\ldots i_l}$, with $\gh{\gamma^{k i_1\ldots i_l}}=2-l$ and $\gh{\beta_{k i_1\ldots i_l}}=3-\gh{\gamma^{k i_1\ldots i_l}}=l+1$, and taking the total $Q$-structure to be 
\begin{multline}
Q'=Q+ \gamma^{I} \bar V_{I}+(-1)^{|K|}\half U^{I}{}_{LK}\gamma^{K} \gamma^{L}\dl{\gamma^{I}}
+
\\
\big((-1)^{|I|}J_{I}-(-1)^{|K|+|I|}U^{L}{}_{IK}\beta_{L}\gamma^{K}\big)\dl{\beta_{I}}\,,
\end{multline}
where $I=(k i_1\ldots i_l)$ is a multi-index labelling the basis elements of $\algg$, $|I|=l-1$ is the degree of $V_I$, $U^{I}{}_{KL}$ are the structure constants of $\algg$ in this basis, $\bar V_{I}$ is the lift of $V_{I}$ to $E$, and $J_{I}$ is the Hamiltonian of $\bar V_{I}$ with respect to $\omega$.

Let us inspect the equations of motion encoded in the deformed system. Because the presymplectic structure is nondegenerate in the sector of $\gamma,\beta$ variables the equations of motion in this sector can be written as
\begin{equation}
\dx\sigma^*(\gamma^I)=\sigma^*(Q'\gamma^I)\,, \qquad
\dx\sigma^*(\beta_I)=\sigma^*(Q'\beta_I)\,.
\end{equation}
In particular, it is instructive to look at the equations associated with the ghosts of degree $0$ and $-1$, which in the basis chosen above are the totally antisymmetric $\gamma^{kij}$ and $\gamma^{kijl}$, respectively. By degree reasoning $\sigma^*(\gamma^{kijl})=0$, while $Q'\gamma^{kij}=0$ because
$\algg$ has no elements of degree lower than $1$. The equations then take the form
\begin{equation}\label{2cond}
\dx\,\sigma^*(\gamma^{kij})=0\,, \qquad
\sigma^*\big(\gamma^{mij}\gamma^{lmk}+\text{cycl.}(ijk)\big)=0\,,
\end{equation}
which state that $\sigma^*(\gamma^{kij})$ are constants defining a Lie algebra structure on $W$ for which $\delta_{ij}$ is an invariant inner product. It is remarkable that the structure (coupling) constants enter this formulation as dynamical fields which are constant on-shell and satisfy the Jacobi identity as a consequence of the equations of motion. At the level of the intrinsic action,
the first equation in~\eqref{2cond} arises from varying the 3-forms
$\beta_{kij}$, while the second one arises from varying the 4-forms
$\beta_{kijl}$, which thus play the role of Lagrange multipliers imposing the Jacobi identity. Indeed, the intrinsic action for $(E',Q',\omega')$ reads as
\begin{equation}
S'=S_0[A,F]+\int_X\Big(\gamma^{kij}\big(\sigma^*(J_{kij})+d\beta_{kij}\big)
+\gamma^{mij}\gamma^{mkl}\,\beta_{ijkl}\Big)\,,
\end{equation}
where $S_0[A,F]$ is the action of $N$ copies of Maxwell theory and we keep the same notation for the fields $\sigma^*(\gamma^{kij})$ (0-forms),
$\sigma^*(\beta_{kij})$ (3-forms) and $\sigma^*(\beta_{kijl})$ (4-forms). Setting
$\gamma^{kij}=\lambda f^k{}_{ij}$ recovers action~\eqref{deformedYM}. Note that $\beta_{i_1\ldots i_m}$ has degree $n-4+m$, so that it is only because $\dim X=4$ that $\sigma^*(\beta_{i_1\ldots i_m})$ is an $m$-form.

The resulting system admits a simple interpretation. Its solutions in the sector of degree-$0$ ghosts are Maurer--Cartan elements of $\algg$. Each of them defines a consistent deformation of $N$ copies of Maxwell theory into the Yang--Mills theory with structure constants $\gamma^{kij}$, and the remaining equations of motion are precisely those of this Yang--Mills theory. In other words, the gauged system is a single local gauge theory which contains the original fields together with all consistent deformations of the above type, the latter entering as global (topological) moduli whose consistency conditions are imposed by the equations of motion. Note that $\algg$ has no elements of degree $0$, so that structure constants related by $O(N)$ rotations are not identified. One can also include quadratic Hamiltonians, i.e.\ gauge the $\mathfrak{o}(N)$ symmetry as well, in which case the moduli are expected to become Lie algebra structures compatible
with $\delta_{ij}$, modulo isomorphism. The only reason for not doing this explicitly in the above consideration is to keep the exposition light.

The above system is closely related to the construction introduced
in~\cite{Grigoriev:2006tt}, see also~\cite{Bekaert:2013zya,Grigoriev:2021bes} for a more refined exposition, where a ghost field is associated with every element of the graded Lie algebra underlying the BFV or BV formulation of a gauge system, and the differential is a generalised Chevalley--Eilenberg one. The resulting $Q$-manifold describes a local gauge theory whose solutions are consistent BRST charges or operators. For instance, starting from a quantum BFV system, the associated gauge theory can be regarded as a field theory of
consistent backgrounds on which the initial system can propagate. A conceptually similar interpretation of string field theory was put forward already in~\cite{Horowitz:1986dta}. The difference with the present construction is that here the sector describing deformations is coupled to the initial system rather than replacing it: the gauged system contains both the original fields and the deformation parameters, and on each solution the original fields obey the equations of the deformed theory selected by the values of these parameters, which are constant on-shell, as befits coupling
constants.

A procedure of a similar nature is the extended antifield formalism
of~\cite{Barnich:1998qz,Barnich:2000me}, which extends the earlier treatment of rigid symmetries in~\cite{Brandt:1997cz} so as to cover deformations and anomalies. There, all local BRST cohomology classes, in particular all consistent deformations of the action, are coupled to the BV master action through constant sources of the appropriate ghost numbers, and the resulting extended action satisfies an extended master equation encoding the algebraic structure of these classes. The gauging of deformations considered in this work can be seen as a local counterpart of this extension: the coupling constants are promoted to fields which, together with their conjugate polymomenta, form a topological sector of a local gauge theory, while their constancy and algebraic constraints such as the Jacobi identity in the above example follow from the equations of motion.

\section{Dualities via topological deformations}\label{sec:dualities}

From the physics perspective, dual field theories are theories that describe the same local degrees of freedom in terms of apparently different (gauge) fields. Typically, dual theories are obtained by using the spacetime Poincar\'e lemma and hence the dual formulations generally differ by some global degrees of freedom which, in turn, are sensitive to spacetime topology. Classical examples include the electric--magnetic and $p$-form dualities~\cite{Deser:1976iy,Henneaux:1986ht}, see also~\cite{Fradkin:1984ai,Freedman:1980us}, and~\cite{Curtright:1980yk,Hull:2001iu,Bekaert:2002dt,Boulanger:2003vs,Henneaux:2004jw,Matveev:2004ac,Boulanger:2012df,Boulanger:2012mq} for more general gauge fields. It turns out that the topological deformation construction presented in Section~\bref{sec:gauging} is well suited to generate dual formulations of a given system. In the terminology of the introduction these are pairs of duality type: the gauged symmetries have nonpositive degree, and the associated momenta give rise to the dual gauge fields.

\subsection{Massless scalar -- $B$-field duality}
\label{sec:scalar-2form-duality}
The standard way to observe duality between the scalar field and 2-form gauge field in 4 dimensions is as follows. Consider the first order action of scalar field:

\begin{equation}
\begin{gathered}
       S = \int d^4x \, \phi_a \partial^a \phi - \frac{1}{2} \phi_a \phi^a\,,
\end{gathered}
\end{equation}
and perform a field redefinition by $F_{abc} = \epsilon_{abcd}\phi^d$, or in reverse $\phi_a = \frac{1}{3!}\epsilon_{abcd}F^{bcd}$. The action takes the form
\begin{equation}
    S = \int d^4x \, \frac{1}{12} F_{abc} F^{abc} - \frac{1}{3!}\phi \epsilon_{abcd} \partial^a F^{bcd}\,,
\end{equation}
and the equation of motion for $\phi$ reads:
\begin{equation}
\begin{gathered}
    \frac{\delta S}{\delta \phi} = -\frac{1}{3!} \epsilon_{abcd} \partial^a F^{bcd} = 0\,.
\end{gathered}
\end{equation}
If de Rham cohomology of $X$ is trivial in degree $3$, this equation has a unique solution modulo gauge equivalence: $F_{abc} = G_{abc} \equiv 3\partial_{[a} B_{bc]}$. The initial action expressed in terms of $B$ reads as
\begin{equation}
\begin{gathered}
    S = \frac{1}{12}\int d^4x \, G_{abc}G^{abc}\,,
    \label{B_action}
\end{gathered}
\end{equation}
and is precisely the standard action for the 2-form gauge field. Let us stress that this "derivation" is not a chain of local equivalences as the passage from $F$ to $B$ is strictly speaking a nonlocal operation. Nevertheless, by inspecting the equations of motion and gauge equivalence it is easy to see that the space of gauge inequivalent solutions may differ by a finite dimensional factor only.
As we are going to see next, the difference becomes apparent if the same steps are reformulated as gauging the global symmetry of the massless field.

The minimal presymplectic gPDE $(E,Q,\omega)$ describing the free massless scalar field was already described in Section~\bref{sec:massive-scalar} and we keep using the same notations for the case $\dim{X}=4$. This presymplectic gPDE has the following degree-$0$ symmetry:
\begin{equation}
\begin{gathered}
     Y = -\frac{\partial}{\partial \phi}, \qquad \gh{Y} = 0\,,
\end{gathered}
\end{equation}
which clearly satisfies the conditions of Proposition~\bref{prop:main} and hence defines a topological deformation. It is given by a new gPDE $(E',Q',\omega')$, where $E' = E \times T^*[3]\mathbb{R}[1]$ with the second factor coordinatized by $c, b$ of ghost degrees $1$ and $2$ respectively. The $Q'$ and $\omega'$ on $E'$ are given by: 
\begin{equation} \label{deformed-scalar}
\begin{gathered}
    Q' \phi = \theta^a \phi_a - c\,, \qquad Q' b = -\theta^{(3)}_a \phi^a\,, \\
    \omega' = -d(\theta^{(3)}_a \phi^a d\phi) +dbdc\,,
\end{gathered}
\end{equation}
and the corresponding deformed Hamiltonian reads as
\begin{equation}
    H' = -\frac{1}{2}\theta^{(4)} \phi^a\phi_a + c \theta^{(3)}_a \phi^a\,.
\end{equation}

Let us analyse what is the field theory the deformed gPDE just constructed describes. The corresponding intrinsic action reads as: 
\begin{equation} \label{gauged-scalar-act}
    S' = \int d^4x \, \phi^a \partial_a \phi - \frac{1}{2}\phi^a\phi_a + A_a(\phi^a +\frac{1}{2}\epsilon^{abcd}\partial_b B_{cd})\,,
\end{equation}
where the sections are parametrised as $\sigma^*(c) = A_a(x) \theta^a\,,$ $\sigma^*(b) =\frac{1}{2}B_{ab}(x)\theta^a \theta^b$ and $\phi(x), \phi_a(x)$ by abuse of notation denote the pullback of the corresponding coordinates by $\sigma$. The only nontrivial gauge transformations are the usual gauge symmetry for 2-form gauge field and the shift gauge symmetry for $\phi$. The above action describes a 2-form gauge field. Indeed, the equivalence can be seen by eliminating the auxiliary field $\phi^a$ and setting $\phi = 0$ by its shift gauge transformation, giving the usual first order action of a 2-form gauge theory. One can also check the equivalence at the level of the associated BV formulations, see Appendix~\bref{app:direc-comp} for details.

To observe the duality directly at the level of presymplectic gauge PDEs, one can eliminate the contractible pair $\phi, c$:
\begin{equation}
\begin{gathered}
    Q' \phi =\theta^a \phi_a - c \Rightarrow \phi = 0, \ c = \theta^a \phi_a\,.
\end{gathered}
\end{equation}
Upon renaming $F_{abc} = \epsilon_{abcd}\phi^d$ this reduction gives the minimal presymplectic formulation for the $B$-field:
\begin{equation}
\begin{gathered}
     Q_{\rm red} F_{abc} = 0\,, \\
     Q_{\rm red} b = \frac{1}{3!} F_{abc}\theta^a \theta^b \theta^c\,,
\end{gathered}
\end{equation}
where $Q_{\rm red}$ denotes the restriction of $Q'$ to the subbundle $E_{\rm red}\subset E'$ singled out by the above conditions. The pullback $\omega_{\rm red}$ of $\omega'$ to $E_{\rm red}$ takes the form 
\begin{equation} \label{def-scal-reduced}
\begin{gathered}
     \omega_{\rm red} = -\frac{1}{3!} d (\theta^{(1)}_{abc}F^{abc} db)\,.
\end{gathered}
\end{equation}
The intrinsic action associated to \eqref{def-scal-reduced} is given by 
\begin{equation}
    S[F_{abc}, B_{ab}] = \int d^4x \Big( \frac{1}{6}F^{abc}G_{abc} - \frac{1}{12} F_{abc} F^{abc}\Big)\,,
\end{equation}
where $\sigma^* b = \frac{1}{2}B_{ab}\theta^a \theta^b$. It is important to note, however, that the above short-cut argument is not a proof because it is not guaranteed that elimination of a contractible pair of coordinates results in the strict equivalence of the corresponding presymplectic gPDEs. Nevertheless, in the case at hand the equivalence is obvious and the argument can be made rigorous by passing to the corresponding BV formulations, see Appendix~\bref{app:direc-comp}.

Let us mention that it is possible to reverse the procedure and obtain the scalar from the B-field by gauging $Y = -\frac{\partial}{\partial b}$ and eliminating contractible pairs. Note also that action~\eqref{gauged-scalar-act} has recently been considered in~\cite{Witten:2026twr} in a more general setting, where $\phi$ and $B$ take values in a generally nontrivial $U(1)$ bundle, and the quantum version of the duality was analysed.

\subsection{$p$-form duality}
The example of the previous subsection
has a straightforward extension to $p$-form gauge theories.
The minimal presymplectic gPDE describing $p$-form gauge field in $n$ dimensions is coordinatised by
\begin{center}
\begin{tabular}{| c | c | c |} 
 \hline
 coordinate & $F^{(p)}_{a_1 \dots a_{p+1}}$ & $C$\\ 
 \hline
 ghost & 0 & p \\ 
 \hline
\end{tabular}
\end{center}
along with the base space coordinates $x^a,\theta^a$. It is convenient to introduce the following notations:
\begin{equation}
\begin{gathered}
    F^{(p)} = \frac{1}{(p+1)!}F^{(p)}_{a_1 \dots a_{p+1}}\theta^{a_1}\dots\theta^{a_{p+1}}\,,
\end{gathered}
\end{equation}
and
\begin{equation}
\begin{gathered}
    \star F^{(p)} = \frac{1}{(n-p-1)!(p+1)!} \epsilon_{a_1 \dots a_{p+1}b_{1} \dots b_{n-p-1}}F^{(p)a_1 \dots a_{p + 1}} \ \theta^{b_{1}} \dots \theta^{b_{n-p-1}}\,.
\end{gathered}
\end{equation}
The $Q$-structure reads as
\begin{equation}
\begin{gathered}
    QC = F^{(p)}\,, \\
    QF^{(p)}_{a_1 \dots a_{p+1}} = 0\,,
\end{gathered}
\end{equation}
and presymplectic structure has the following form:
\begin{equation}
\begin{gathered}
    \omega = (-1)^{n(p+1)+1}d(\star F^{(p)} dC)\,.
\end{gathered}
\end{equation}
For the following parametrization of sections $\sigma^*F^{(p)}_{a_1\dots a_{p+1}}= F^{(p)}_{a_1\dots a_{p+1}}(x), \, \sigma^*C = \frac{1}{p!}A_{a_1 \dots a_p} \theta^{a_1} \dots \theta^{a_p}$, and defining the curvature $G_{a_1\dots a_{p+1}} = (p+1)\partial_{[a_1}A_{a_2\dots a_{p+1}]}$, we get the following intrinsic action:
\begin{equation}
    S(F^{(p)}, A^{(p)}) = \frac{1}{(p+1)!}\int d^nx \, \Big((-1)^p F^{a_1 \dots a_{p+1}} G_{a_1 \dots a_{p+1}} -\frac{1}{2}F_{a_1 \dots a_{p+1}}F^{a_1 \dots a_{p+1}}\Big)\,.
\end{equation}
See \cite{Dneprov:2024cvt} for further details.

The above presymplectic gPDE has an obvious degree $-p$ symmetry $V = -\frac{\partial}{\partial C}$ obeying 
\begin{equation}
\begin{gathered}
    [Q,V] = 0, \, L_V \omega \in \mathcal{I}\,.
\end{gathered}
\end{equation}
The corresponding conserved current takes the form:
\begin{equation}
\begin{gathered}
    J_V = - \star F^{(p)}\,.
\end{gathered}
\end{equation}

Straightforward application of the gauging procedure described by Proposition~\bref{prop:main} gives the topological deformation $(E',Q',\omega')$ where $E'$ is extension of $E$ by new fibre coordinates $c, b$ of degrees $p+1, \ n-p-2$. The modified $Q$-structure reads
\begin{equation}
\begin{gathered}
    Q'C =F^{(p)} - c\,, \\
    Q'F^{(p)}_{a_1 \dots a_{p+1}} = 0\,, \\
    Q' b = - (-1)^{p(n+1)}\star F^{(p)}\,,
\end{gathered}
\end{equation}
and the presymplectic structure is given by: 
\begin{equation}
\begin{gathered}
    \omega' = (-1)^{n(p+1)+1}d(\star F^{(p)} dC) + db dc\,.
\end{gathered}
\end{equation}

To analyse the resulting deformed system one could explicitly consider the corresponding action and gauge transformations determined by $(E',Q',\omega')$ and find that it is equivalent to the theory of the $(n-p-2)$-form. However, just like in the example of the previous subsection there is a very efficient shortcut which is to perform the reduction at the level of presymplectic gPDE.
More precisely, eliminating the contractible pair formed by $C$ and $Q'C\equiv F^{(p)}- c$ one arrives at the subbundle $E_{\rm red}\subset E'$ singled out by 
\begin{equation}
\begin{gathered}
    C = 0, \ c = F^{(p)}\,.
\end{gathered}
\end{equation}
The restriction $Q_{\rm red}$ of $Q'$ and the pullback $\omega_{\rm red}$ of $\omega'$ to this subbundle give all the structure of presymplectic gPDE which can be shown to be equivalent to the starting point one. Again to give a rigorous proof of equivalence one should resort to the corresponding BV formulation. 

To analyse the reduced system, let us introduce new coordinates $F^{(n-p-2)}_{a_1\cdots a_{n-p-1}}$ via
\begin{equation}
\begin{gathered}
   F^{(n-p-2)}\equiv - (-1)^{p(n+1)} \star F^{(p)}\,.
\end{gathered}
\end{equation}
Rewriting $Q_{\rm red}$ and $\omega_{\rm red}$ in terms of new coordinates one finds:
\begin{equation}
\begin{gathered}
    Q_{\rm red}F^{(n-p-2)}_{a_1 \dots a_{n-p-1}} = 0\,, \\
    Q_{\rm red} b = F^{(n-p-2)}\,,
\end{gathered}
\end{equation}
and 
\begin{equation}
\begin{gathered}
    \omega_{\rm red} = (-1)^{n(r+1)+1} d(\star F^{(n-p-2)} db)\,,
\end{gathered}
\end{equation}
where $r=n-p-2$. This is precisely the minimal presymplectic gPDE for the $r$-form gauge field.

Let us finally note that our approach to $p$-form dualities shares some similarity with the frame-like description of higher-spin gauge fields and their dualities~\cite{Boulanger:2012mq}, see also~\cite{Matveev:2004ac}. Another somewhat related approach worth mentioning is that of~\cite{Abakumova:2023wve}.

\subsection{Freedman--Townsend -- PCM duality}

Let us recall basic facts about the principal chiral model and its duality to Freedman--Townsend theory. The Principal Chiral Model (PCM for short) is defined in the following way:
Let $g:X\rightarrow G$ be a map from space-time manifold $X$ to a compact matrix Lie group. Then PCM is defined by the following action functional:
\begin{multline} 
    S =\frac{1}{2} \tr\int d^n x \, (g^{-1}\partial_a g g^{-1}\partial^a g) =\frac{1}{2} \tr\int d^n x \, j^a j_a =\\ \frac{1}{2}\tr\int d^n x \ \ \bar j^a \bar j_a\,,
    \label{PCM}
\end{multline}
where for simplicity we took the spacetime to be the Minkowski space. The currents are defined as follows:
\begin{equation} \label{lr_currents}
\begin{gathered}
    j^a = g^{-1}\partial^a g\,, \qquad \bar j^a = \partial^a g\, g^{-1}\,.
\end{gathered} 
\end{equation}
The equations of motion take the form
\begin{equation}
\begin{gathered}
    \partial_{a}(g^{-1} \partial^a g) = \partial_a j^a = \partial_a \bar j^a = 0\,,
\end{gathered} 
\end{equation}
and the model possesses the following global $G \times G$ symmetry:
\begin{equation}
\begin{gathered}
    g \rightarrow h_L g h_R, \, (h_L,h_R)\in G \times G\,.
\end{gathered} 
\end{equation}
The currents \eqref{lr_currents} are the corresponding Noether currents: $j^a$ is left-invariant and is associated with the right action $g\mapsto gh_R$, while $\bar j^a$ is right-invariant and is associated with the left action $g\mapsto h_Lg$.

The action of the Freedman--Townsend (FT) theory \cite{Freedman:1980us} is given by 
\begin{equation}
    S[B,A] = \int_X \tr\Big( B \wedge F(A) + \frac{1}{2}A\wedge * A \Big)\,,
    \label{freedman_town_first}
\end{equation}
where $X$ is a 4-dimensional spacetime, $B$ and $A$ are $2$- and $1$-forms on $X$ with values in the (matrix) Lie algebra $\algg$, and $\tr$ is an invariant nondegenerate trace on $\algg$. 

The usual derivation of the duality between FT and PCM proceeds as follows. The equations of motion for $B$ imply that the curvature $F(A)\equiv dA+A^2$ vanishes. For simply connected $X$ this implies that $A$ can be solved via:
\begin{equation}
    \begin{gathered}
        A = g^{-1}dg\,.
    \end{gathered}
\end{equation}
By substituting this form of $A$ in the action \eqref{freedman_town_first}, one obtains the action \eqref{PCM}.

The PCM--FT duality is the nonabelian version of the massless scalar -- 2-form gauge field duality discussed in Section~\bref{sec:scalar-2form-duality}. The natural question is whether this duality can be systematically obtained via gauging a global symmetry of the underlying presymplectic gPDE. To answer this we first need a concise presymplectic BV-AKSZ formulation of PCM.

\subsubsection{Presymplectic formulation for PCM}

Restricting ourselves to the case where $X$ is the $n$-dimensional Minkowski space and working in Cartesian coordinates $x^a$ the minimal presymplectic gPDE formulation of PCM is constructed as follows: the fibre $W$ of $E=T[1]X\times W$ is given by $W=G\times (\algg \tensor \fR^n)$, where $\fR^n$ is considered as a cotangent space
at any point of $X$. We use matrix-valued $g$ as coordinates on the group and $\algg$-valued coordinates $P_a$ on $\algg \tensor \fR^n$, which we also treat as matrices. Here $W$ is identified with $TG\tensor\fR^n$ via the left-invariant Maurer--Cartan form $g^{-1}dg$. The $Q$-structure is given by: 
\begin{equation}
\begin{gathered}
    Q g=\theta^a g P_a\,, 
        \qquad 
    QP_a = \frac{1}{2}\theta^d[P_a, P_d]\,,
\end{gathered}    
\end{equation}
and the presymplectic structure reads as:
\begin{equation}
\begin{gathered}
    \omega = d\chi = (-1)^{n-1}d(\theta^{(n-1)}_a \tr(P^a g^{-1}dg))\,,
\end{gathered}    
\end{equation}
where $\tr$ is the invariant trace. It is easy to check that $Q$ is Hamiltonian and the Hamiltonian of $Q$ is given explicitly by 
\begin{equation}
   H=- \frac{1}{2} \theta^{(n)} \tr(P_a P^a)\,, \qquad   
   i_Q \omega +dH \in \cI\,.
\end{equation}
The corresponding intrinsic action is given by: 
\begin{equation}
\begin{gathered}
    S = \tr \int d^nx \, P^a g^{-1}\partial_a g -\frac{1}{2}P_a P^a\,,
\end{gathered} 
\end{equation}
and is immediately identified as the first-order formulation of PCM. It is easy to see that $P_a=g^{-1}\partial_ag$ on solutions in accord with the identification $TG \simeq G \times \algg$ made above.

\subsubsection{Duality}

To consider the PCM -- FT duality we set $n=4$
and observe that the following vertical vector fields on $E$
\begin{equation}
V_i g = gT_i\,, \qquad V_i P_a= [P_a, T_i]\,,
\end{equation}
where $T_i$ form a basis of $\algg$,
represent the internal symmetries of the presymplectic gPDE underlying PCM. Indeed, it is easy to check that 
\begin{equation}
\begin{gathered}
    i_{V_i} \omega = -dJ_i = -d \, \tr(\theta^{(3)}_a P^a T_i)\,.
\end{gathered} 
\end{equation}
Moreover, the classical anomaly is absent
\begin{equation}
\begin{gathered}
    V_i J_j - f^k{}_{ij} J_k = \tr(\theta^{(3)}_a [P^a, T_i] T_j ) - f^k{}_{ij} \tr(\theta^{(3)}_a P^a T_k ) = 0\,,
\end{gathered} 
\end{equation}
and hence all the conditions of Proposition~\bref{prop:main} are satisfied. Note that $V_i$ are the left-invariant vector fields on $G$, extended to $E$. They satisfy $[V_i,V_j]=f^k{}_{ij}V_k$ and generate the right action $g\mapsto gh$ of $G$, under which $P_a$ transforms in the adjoint representation, and $J_i$ are the corresponding Noether currents, built from $j_a$, cf.~\eqref{lr_currents}. 

The result of gauging is the presymplectic gPDE $E' = E\times \algg[1] \times \algg^*[2]$ with extra coordinates $c, \ \gh{c} = 1$ and $b, \ \gh{b} = 2$ valued in $\algg$ (with $\algg^*$ identified with $\algg$ via $\tr$) and parametrizing $\algg[1] \times \algg^*[2]$. The extended presymplectic structure $\omega'$ and $Q'$ are given by:
\begin{equation}
\begin{gathered}
    \omega' = d\chi'\,, \qquad \chi'=-\tr\big(\theta^{(3)}_a P^a g^{-1}dg\big) + \tr(b\,dc)\,,
\end{gathered} 
\end{equation}
and
\begin{equation}
\begin{gathered}
    Q'c = -\frac{1}{2}[c,c]\,, \\
    Q' b= \theta^{(3)}_a P^a - [c,b]\,,
\end{gathered} 
\end{equation}
while the action of $Q'$ on the $E$ part is given by $Q+c^iV_i$, where $c=c^iT_i$, e.g.
\begin{equation}
Q'g=g\theta^a P_a + gc \,, \qquad 
    Q'P_a = \frac{1}{2} \theta^d[P_a, P_d] + [P_a, c]\,.
\end{equation}
The associated intrinsic action reads as:
\begin{multline}
    S = \tr\int d^4x \ -\frac{1}{2}P_a P^a + P^a g^{-1}\partial_a g - A_a P^a + 
    \\
    \frac{1}{4}\epsilon^{abcd} B_{ab}(\partial_c A_d - \partial_d A_c + [A_c, A_d])\,,
\end{multline} 
where $\sigma^*(c)=A_a\theta^a$ and $\sigma^*(b)=\half B_{ab}\theta^a\theta^b$. Eliminating $P_a$ through its equation of motion $P_a=g^{-1}\partial_a g-A_a$ and using the gauge freedom to set $g=1$, at least locally, one recovers the FT action~\eqref{freedman_town_first}.

Just like in the preceding examples, an efficient shortcut to derive FT theory from the above topological deformation is to eliminate contractible pairs at the level of presymplectic gPDE. Namely, $g$ and $c' = g^{-1}Q'g=\theta^a P_a + c$ is a contractible pair, at least locally. Ignoring the global geometry of $G$ the resulting presymplectic gPDE $(E_{\rm red},Q_{\rm red},\omega_{\rm red})$, where $E_{\rm red}\subset E'$ is the subbundle obtained by eliminating this pair, is determined by:
\begin{equation}
\begin{gathered}
    Q_{\rm red}b = \theta^{(3)}_a P^a + \theta^a[P_a,b]\,, \\
    Q_{\rm red}P_a = -\frac{1}{2}\theta^d [P_a, P_d]\,,
    \label{min_freedman_town}
\end{gathered} 
\end{equation}
and 
\begin{equation}
    \omega_{\rm red} = -d\, \tr(\theta^{a} P_a db)\,.
\end{equation}
This coincides with the presymplectic gPDE formulation of the Freedman--Townsend theory put forward in~\cite{Dneprov:2024cvt}.

Finally, let us mention a possible relation to covariant actions for duality-symmetric and chiral (self-dual) fields, where a field and its dual enter a single action. A prominent example is the Pasti--Sorokin--Tonin approach~\cite{Pasti:1995tn,Pasti:1996vs,Pasti:1997gx}, where manifest Lorentz invariance is achieved by introducing an auxiliary scalar field which is pure gauge. Alternative formulations were proposed by Sen~\cite{Sen:2015nph,Sen:2019qit} and further developed in~\cite{Hull:2023dgp,Hull:2026osk}. In these formulations, the chiral field is accompanied by an additional sector that decouples from the physical one, which is reminiscent of the decoupled topological sector produced by the gauging considered here. It would be interesting to understand whether such formulations can be obtained or reinterpreted within the present framework, e.g.\ by gauging suitable symmetries of the presymplectic gPDE describing a pair of dual theories. It is also worth mentioning another approach employing the embedding into a topological theory in higher dimensions~\cite{Arvanitakis:2022bnr,Evnin:2023ypu,Chen:2026chi}.

\section*{Acknowledgments}
MG is grateful to G. Barnich and E. Skvortsov for fruitful discussions. Useful exchanges with S. Lyakhovich are acknowledged. The work of ID and MG is supported by the ULYSSE Incentive Grant for Mobility in Scientific Research [MISU] F.6003.24, F.R.S.-FNRS, Belgium. We used the AI model Claude (Anthropic) to assist with editing the text, unifying terminology and notation, and cross-checking signs and coefficients in several computations. All results were verified by the authors, who take full responsibility for the content of the paper. 

\appendix

\section{Scalar -- $B$-field equivalence in BV terms} \label{app:direc-comp}

Consider the gauged system $(E',Q',\omega')$ of Section~\bref{sec:scalar-2form-duality}. We introduce coordinates on the space $S\Gamma(E')$ of supersections via
\begin{equation}
    \begin{gathered}
        \ev^* \phi = \phi + \phi_{|a}\theta ^a + \dots\,, \\
        \ev^* \phi^a = \phi^a + \phi^a_{|b}\theta^b + \dots\,, \\
        \ev^* c = C + C_{|a} \theta^a + \frac{1}{2}C_{|ab} \theta^a \theta^b + \frac{1}{3!}C_{|abc} \theta^a \theta^b \theta^c + \frac{1}{4!}C_{|abcd} \theta^a \theta^b \theta^c \theta^d\,, \\
        \ev^* b = B + B_{|a} \theta^a + \frac{1}{2}B_{|ab} \theta^a \theta^b + \frac{1}{3!}B_{|abc} \theta^a \theta^b \theta^c + \frac{1}{4!}B_{|abcd} \theta^a \theta^b \theta^c \theta^d\,.
    \end{gathered}
\end{equation}

For the components not in the kernel of the induced presymplectic structure we use the following names:
\begin{equation}
\begin{aligned}
    \phi &= \phi\,, &\qquad \phi^* &= -\phi_{|a}^{a}\,, \\
    \phi^a &= \phi^a\,, &\qquad \phi^*_a &= \phi_{|a}\,, \\
    A_a &= C_{|a}\,, &\qquad A^{*a} &= \frac{1}{3!}\epsilon^{abcd} B_{|bcd}\,, \\
    B_{ab} &= B_{|ab}\,, &\qquad B^{*ab} &= \frac{1}{2}\epsilon^{abcd} C_{|cd}\,, \\
    C &= C\,, &\qquad C^* &= \frac{1}{4!}\epsilon^{abcd}B_{|abcd}\,, \\
    \zeta_a &= B_{|a}\,, &\qquad \zeta^{*a} &= -\frac{1}{3!}\epsilon^{abcd} C_{|bcd}\,, \\
    \xi &= B\,, &\qquad \xi^* &= \frac{1}{4!}\epsilon^{abcd} C_{|abcd}\,,
\end{aligned}
\end{equation}
which induces the following BV symplectic structure:
\begin{multline}
    \Omega = \int_X d^4x \Big( \delta\phi^* \delta\phi + \delta\phi^*_a \delta\phi^a + \delta A^{*a} \delta A_a +\frac{1}{2}\delta B^{*ab}\delta B_{ab} \\
    + \delta C^*\delta C + \delta\zeta^{*a} \delta \zeta_a +\delta \xi^* \delta \xi \Big)\,.
\end{multline}

The associated BV action is given by:
\begin{multline}
    S_{BV} = \int d^4x \Big(\phi^a \partial_a \phi - \frac{1}{2}\phi^a\phi_a + A_a\big(\phi^a + \frac{1}{2}\epsilon^{abcd}\partial_b B_{cd}\big) +\phi^*C - A^{*a} \partial_a C \\ - B^{*ab} \partial_{a} \zeta_{b} + \zeta^{*a} \partial_a \xi\Big)\,.
    \label{B_parenta}
\end{multline}

We can eliminate the group of generalized auxiliary fields $\phi, \phi^*$, $C, C^*$:
\begin{equation}
\begin{gathered}
    \frac{\delta S}{\delta \phi^*} \bigg |_{C^* = \phi = 0} = C = 0\,, \\
    \frac{\delta S}{\delta C}\bigg |_{C^* = \phi = 0} = \phi^* + \partial_a A^{*a} = 0\,,
\end{gathered}
\end{equation}
and $\phi^a, \phi^*_a$:
\begin{equation}
        \frac{\delta S}{\delta \phi^a}\bigg|_{\phi = 0} = A_a - \phi_a = 0\,.
\end{equation}

The resulting BV action reads
\begin{equation}
\begin{gathered}
    S_{BV} = \int d^4x \Big(\frac{1}{2}A^aA_a + \frac{1}{2}\epsilon^{abcd} A_a\partial_b {B}_{cd} - B^{*ab} \partial_{a} \zeta_{b} + \zeta^{*a} \partial_a \xi\Big)\,.
    \label{B_parent}
\end{gathered}
\end{equation}

This is the first-order formulation of the $2$-form gauge theory, where $A^a$ plays the role of the dual field strength. Eliminating $A_a$ and $A^{*a}$, using $A^a=-\frac{1}{3!}\epsilon^{abcd}G_{bcd}$, gives action~\eqref{B_action} for the $B$-field.

\section{Conventions}\label{app:conventions}
Our conventions follow those of~\cite{Alkalaev:2013hta,Grigoriev:2024ncm}, see also~\cite{Ikeda:2012pv}. We collect the relevant formulas below. We use the mostly minus convention for the metric, $\eta_{ab} = diag(+,-,-,-)$. Levi-Civita symbol is defined as follows:
\begin{equation}
    \epsilon_{01\dots n-1} = \epsilon^{01\dots n-1} = 1\,.
\end{equation}
We denote the parity of an object by $| \ |$. Differential forms are bigraded by ghost degree and form degree; $d$ has ghost degree $0$ and form degree $1$. Signs are determined by the total parity, i.e.\ ghost degree plus form degree mod $2$, so that in particular $dz^A$ has parity $|z^A|+1$. The forms will be generally written with the differentials on the left, when not explicitly specified otherwise:
\begin{equation}
    \alpha = dz^{A_1} \cdots dz^{A_k} \alpha_{A_1 \dots A_k}\,,
\end{equation}
and, accordingly,
\begin{equation}
    df = dz^A \frac{\partial f}{\partial z^A}\,.
\end{equation}
The insertion of a vector field into a differential form is given by
\begin{equation}
    i_X = X^A \frac{\partial}{\partial dz^A}\,,
\end{equation}
where $\partial/\partial z^A$ and $\partial/\partial dz^A$ denote left derivatives. The evaluation of a $1$-form on a vector field is given by
\begin{equation}
    \alpha(X) = i_X \alpha\,.
\end{equation}
In particular, $(\sigma^*\chi)(\dx)=i_{\dx}\sigma^*\chi$ in the intrinsic action. The Lie derivative is defined by
\begin{equation}
    L_X = i_X d +(-1)^{|X|}d i_X\,.
\end{equation}

For a graded manifold $M$ with local coordinates $x^A$, $T[k]M$ denotes the shifted tangent bundle, whose fibre coordinates $\xi^A$ have degree $\gh{x^A}+k$. In particular, $T[1]X$ has coordinates $x^a,\theta^a$ with $\gh{\theta^a}=1$, and $\dx=\theta^a\partial/\partial x^a$. Similarly, $T^*[k]M$ denotes the shifted cotangent bundle whose canonical symplectic form has degree $k$: if $x^A$ are local coordinates on $M$, the conjugate fibre coordinates $p_A$ have degree $k-\gh{x^A}$. Similarly, for $W$ a real linear space, $W[p]$ denotes the graded manifold with linear coordinates of degree $p$, so that
$T^*[k](W[p]) \cong W[p]\times W^*[k-p]$.

\providecommand{\href}[2]{#2}\begingroup\raggedright\endgroup

\end{document}